\documentclass{CSML}

\def\dOi{11(4:5)2015}
\lmcsheading%
{\dOi}
{1--31}
{}
{}
{Dec.~\phantom09, 2014}
{Nov.~10, 2015}
{}

\ACMCCS{[{\bf Theory of computation}]: ???}

\usepackage{etex}
\newif\iflncs\lncsfalse
\newif\iffull
\iflncs\fullfalse\else\fulltrue\fi
\newif\ifmai\maifalse

\usepackage[T1]{fontenc}
\usepackage{lmodern}

\usepackage[mathscr]{euscript}%
\usepackage{enumitem,amsmath, amsfonts,amstext,amssymb,
verbatim}
\usepackage[english]{babel}

\usepackage[all]{xy}
\usepackage{float}
\usepackage{latexsym}
\usepackage{url}

\usepackage{pgf}
\usepackage{tikz}
\usepackage{wrapfig}
\usepackage{fancyhdr}
\usepackage{multicol}
\usepackage{xcolor}
\usepackage{algpseudocode}
\usepackage{algorithm}
\usepackage{graphicx}
\usepackage[font=small]{caption}
\usepackage{subcaption}
\usetikzlibrary{arrows,plotmarks,shapes,snakes,automata,backgrounds,petri,fit}
\usepackage{array}

\newcounter{myenum}
  {\end{list}}

\newif\ifmai\maifalse
\newif\ifsp\iflncs\spfalse\else\sptrue\fi

\renewcommand{\arraystretch}{1.5}

\ifmai
\iflncs{}\else \fi
\iflncs{}\else \fi

\iflncs{}\else \fi

\iflncs{}\else  \fi

\iflncs{}\else \fi

\iflncs{}\else \fi
\iflncs{}\else \fi
\fi

\newcommand{\defin}{\stackrel{\triangle}{=}}

\makeatletter
\def \rightarrowfill{\m@th\mathord{\smash-}\mkern-6mu%
  \cleaders\hbox{$\mkern-2mu\mathord{\smash-}\mkern-2mu$}\hfill
  \mkern-6mu\mathord\rightarrow}
\makeatother

\makeatletter
\def \mapstofill{\shortmid\!\!\!\!\m@th\mathord{\smash-}\mkern-6mu%
  \cleaders\hbox{$\mkern-2mu\mathord{\smash-}\mkern-2mu$}\hfill
  \mkern-6mu\mathord\rightarrow}
\makeatother

\newcommand{\dom}{\mathrm{dom}}

\newcommand{\mik}[1]{}
\newcommand{\fra}[1]{}

\newcommand{\es}{\emptyset}

\makeatletter
\def \rightarrowfill{\m@th\mathord{\smash-}\mkern-6mu%
  \cleaders\hbox{$\mkern-2mu\mathord{\smash-}\mkern-2mu$}\hfill
  \mkern-6mu\mathord\rightarrow}
\makeatother

\makeatletter
\def \longrightarrowfill{\m@th\mathord{\smash-}\mkern-6mu%
  \cleaders\hbox{$\mkern-2mu\mathord{\smash-}\mkern-2mu$}\hfill
  \mkern-6mu\mathord\longrightarrow}
\makeatother

\makeatletter
\def \Rightarrowfill{\m@th\mathord{\smash=}\mkern-6mu%
  \cleaders\hbox{$\mkern-2mu\mathord{\smash=}\mkern-2mu$}\hfill
  \mkern-6mu\mathord\Rightarrow}
\makeatother

\makeatletter
\def \midrightarrowfill{\m@th\mathord{\smash{\raisebox{.2ex}{$\scriptscriptstyle\mid$}}\!\!\,-}\mkern-6mu%
  \cleaders\hbox{$\mkern-2mu\mathord{\smash-}\mkern-2mu$}\hfill
  \mkern-6mu\mathord\rightarrow}
\makeatother

\makeatletter
\def \midRightarrowfill{\m@th\mathord{\smash{\raisebox{.1ex}{$\scriptstyle\mid$}}\!\!\!=}\mkern-6mu%
  \cleaders\hbox{$\mkern-2mu\mathord{\smash=}\mkern-2mu$}\hfill
  \mkern-6mu\mathord\Rightarrow}
\makeatother

\newcommand{\tld}{\tilde}

\renewcommand{\mathring}[1]{\overset\circ{#1}}

\newcommand{\Act}{\mathcal{A}}

\newcommand{\supp}{\mathrm{supp}}

\newcommand{\K}{\mathcal{K}}

\newcommand{\vsp}{\vspace*{.2cm}}
\newcommand{\St}{\mathcal{S}}

\renewcommand{\P}{\mathcal{P}}

\newlength{\algrhswidth}
{\left\lbrace\begin{array}{@{}l@{}}}%
{\end{array}\right.}

\newcommand{\Y}{\mathcal{Y}}
\newcommand{\X}{\mathcal{X}}

\newcommand{\M}{\mathcal{M}}

\newcommand{\range}{\mathrm{range}}

\newcommand{\var}{\mathrm{var}}

\newcommand{\qif}{{\sc  qif}}

\newcommand{\UU}{\hat B}
\newcommand{\HH}{\hat H}
\newcommand{\DD}{\hat D}

\newcommand{\unif}{\mathrm{u}}   
\newcommand{\mdp}{{\sc mdp}}

  \makeatletter 
  \let\@copyrightspace\relax
  \makeatother

\keywords{Quantitative information flow, adaptive attackers, information theory, Bayesian decision theory.}

\ACMCCS{[{\bf Security and privacy}]:  Formal methods and theory of security -- Formal security models. } 

\usepackage{hyperref}
\theoremstyle{plain} 

\begin{document}

\title[Quantitative Information Flow]{Quantitative Information Flow under Generic\\Leakage Functions and Adaptive Adversaries\rsuper*}
\titlecomment{{\lsuper*}Extended and revised version of \cite{BP14}. Corresponding author: Michele Boreale.
Work partially supported by the {\sc eu}  project {\sc Ascens} under the {\sc fet} open initiative in  {\sc fp7} and by the italian {\sc prin} project {\sc cina}.}

\author[M.~Boreale]{Michele Boreale}	
\address{Universit\`a di Firenze, Dipartimento di Statistica, Informatica, Applicazioni (DiSIA), Viale Morgagni 65, I-50134  Florence, Italy.}	
\email{michele.boreale@unifi.it}  

\author[F.~Pampaloni]{Francesca Pampaloni}	
\address{Independent researcher, Italy.}	
\email{francesca.pampaloni@ge.com}  





\begin{abstract}
We put forward a  model of \emph{action-based} randomization mechanisms to analyse quantitative information flow (\qif) under generic leakage functions, and under possibly adaptive adversaries. This model subsumes many of the \qif\ models proposed so far. Our main contributions include the following: (1) we identify mild general conditions on the leakage function  under which it is possible to derive general and significant results on adaptive \qif; (2) we contrast the efficiency of adaptive and non-adaptive strategies, showing that the latter are as efficient as the former in terms of length  up to an expansion factor bounded by the number of available actions; (3) we show that the maximum information leakage over strategies,   given a finite time horizon, can be expressed in terms of a Bellman equation. This can be used to compute an optimal finite strategy recursively, by resorting to standard methods like backward induction.
\end{abstract}

\maketitle
\section{Introduction}\label{sec:Intro}
\emph{Quantitative Information Flow} (\qif) is a well-established approach to  confidentiality analysis: the basic   idea is measuring  how much  information flows  from sensitive
to observable data, relying on tools from Information Theory \cite{CHM01,BK08,CPP08,Smith09,BCP09,Bor09,KS10,BPP11,BPP11b}.
\fra{Ho modificato sia i titoli (tutti in maiuscole) che il modo di scrivere le citazioni. Manca da vedere se la bibliografia contiene tutte le info richieste}

Two important issues that arise in \qif\ are: what measure one should adopt to quantify the leakage of confidential data, and the relationship between adaptive and non adaptive adversaries.  Concerning the first issue, a long standing debate in the \qif\ community concerns the relative merits of leakage  functions based on   Shannon entropy (see e.g. \cite{CHM01,Bor09}) and min-entropy (see e.g. \cite{Smith09,BCP09}); other types of entropies are sometimes considered (see e.g. \cite{KB07}.) As a matter of fact, analytical results for each of these  types of leakage functions have been so far  worked out  in a non-uniform, ad hoc fashion.

Concerning the second issue, one sees that, with the notable exception  of \cite{KB07}  which we discuss later on,  the treatment of confidentiality  in \qif\ has  so far been almost exclusively confined   to   attackers that can only  passively   eavesdrop on the mechanism; or,  at best,  obtain answers in response to   queries (or \emph{actions})  submitted in   a non-adaptive fashion \cite{BP12}. By \emph{passive}, we mean here attackers that can eavesdrop on messages,  but not interfere  with the generation process of these messages\footnote{Passivity in this sense does not rule out attackers that can try secrets, based on some form of oracle. Active attackers are also considered in approaches to quantitative \emph{integrity}  \cite{Clarkson}, a theme that will not be considered here.}.  Clearly, there are situations where this model is not adequate. To mention but two: chosen plaintext/ciphertext attacks against cryptographic hardware or software;   adaptive  querying of   databases whose records contain  both sensitive and non-sensitive  fields.


In this paper, we tackle both   issues outlined above. We: (a)  put forward a general \qif\ model  where the leakage function is built around a \emph{generic uncertainty measure}; and, (b)  derive several general results on the relationship between  adaptive and non adaptive adversaries in this model. In more detail, we assume that, based on a secret  piece of information $X\in \X$,  the mechanism responds to a sequence of queries/actions $a_1,a_2,\ldots$ ($a_i\in Act$), adaptively submitted by an adversary, thus  producing a sequence of  answers/observations $Y\in \Y^*$. Responses  to   individual queries are in general probabilistic, either because of the presence of noise or by system's design. Moreover, the mechanism is stateless, thus answers are independent from one another.  The adversary is assumed to know   the distribution according to which $X$ has been generated (the prior) and the input-output behaviour of the mechanism.
An adaptive adversary can choose the next query based on past observations, according to a predefined strategy.
Once a strategy and a prior have been fixed, they together induce a probability space over   sequences of observations. Observing a specific sequence   provides the adversary with information  that   modifies his belief about $X$,   possibly reducing his uncertainty.  We measure   information leakage as the  \emph{average reduction in uncertainty}. We work with a  {generic} measure of uncertainty, $U(\cdot)$. Formally, $U(\cdot)$ is     just a real-valued function over the set of probability distributions on $\X$, which represent     possible beliefs of the adversary. Just two   properties are assumed of $U(\cdot)$: concavity and continuity. Note that   leakage functions commonly employed in   \qif, such as Shannon entropy, guessing entropy and error probability -- the additive version of Smith's min-entropy-based \emph{vulnerability}   \cite{Smith09} -- do fall in this category.

The other central theme of our study is the  comparison between      adaptive and the simpler non-adaptive strategies. All in all, our results indicate that, for even moderately powerful adversaries, there is no dramatic difference between the two, in terms of difficulty of analysis. A more precise account of our contributions follows.
\begin{itemize}
\item  We put forward a general model of adaptive \qif; we  identify mild general conditions on the uncertainty  function  under which it is possible to derive general and substantial results on adaptive \qif.
\item We compare the difficulty of analyzing mechanisms under adaptive and non-adaptive adversaries. We first note that,  for the class of mechanisms admitting a ``succinct'' syntactic  description  - e.g. devices specified by boolean formulae -  the analysis problem is  intractable ({\sc np}-hard), even if limited to very simple instances of the \emph{non-adaptive} case. This essentially depends on the fact that such mechanisms can feature   exponentially many actions in the syntactic size of the description. In the general case, we show that non-adaptive finite strategies are as efficient as adaptive ones, up to an  {expansion factor} in their length bounded by the number of distinct actions available.  Practically, this indicates that, for mechanisms described in explicit form (e.g. by tables, like a {\sc db}) hence      featuring an ``affordable'' number of  actions available to the adversary, it may be sufficient to assess   resistance of the mechanism against non-adaptive strategies. This is important, because simple analytical results are available   for such   strategies \cite{BP12}.
\item We show that the maximum leakage over all strategies is the same for both adaptive and non-adaptive adversaries, and only depends on an indistinguishability equivalence relation over the set of secrets.
  \item  
       We show that   maximum  leakage over all strategies over a finite horizon can be expressed in terms of a Bellman equation. This equation can be used to compute optimal finite strategies recursively. As an example, we show how to do that using Markov Decision Processes (\mdp's)  and   backward induction.
  \item We finally give a Bayesian decision-theoretic justification of our definition of uncertainty function.  We  argue that each such function arises as a measure of    expected  \emph{loss}; and, vice-versa, that (under a mild condition) each measure of expected loss is in fact an uncertainty function in our sense.
\end{itemize}




\subsection*{Structure of the Paper} Section \ref{sec:model} introduces the model. This is illustrated with a few examples in Section \ref{sec:examples}. The subsequent  four sections 4, 5, 6, 7 discuss the results outlined in (2), (3), (4) and (5) above, respectively. Section \ref{sec:concl} contains a few   concluding remarks, discussion of related work and some directions for further research. Some technical material has been confined to  three separate appendices.

\section{A model of adaptive \qif}\label{sec:model}
\subsection{Randomization mechanisms, uncertainty, strategies}

\begin{defi}
{An \emph{action-based randomization mechanism}\footnote{The term \emph{information hiding system}   is sometimes found in the literature to indicate randomization mechanisms. This   term, however, is also used with a  rather  different technical meaning in the literature on watermarking; so we prefer to avoid it altogether here.}   is a 4-tuple \[\St= (\X, \Y, Act, \{M_a : a\in Act\})\,,\]
where (all sets finite and nonempty):
$\X, \Y$ and $Act$ are respectively the sets of \emph{secrets}, \emph{observations}
and \emph{actions} (or \emph{queries}) and for each $a\in Act$, $M_a$ is a stochastic matrix of dimensions $|\X| \times |\Y|$.
}
\end{defi}

For each action $a\in Act$, $x\in \X$ and $y \in \Y$, the element of row $x$ and column $y$ of $M_a$ is denoted by $p_a(y|x)$. Note that for each $x$ and $a$, row $x$ of $M_a$ defines a probability distribution over $\Y$, denoted by
$p_a(\cdot|x)$.
A mechanism $\St$ is \emph{deterministic} if each entry of each $M_a$ is either 0
or 1. Note that to any deterministic mechanism there corresponds   a function $f:\X\times Act\rightarrow \Y$ defined by $f(x,a)=y$, where $p_a(y|x)=1$. Recall that a real function $F$ defined over a convex  subset $C\subseteq \mathbb{R}^n$ is \emph{concave} if, for each $\lambda\in [0,1]$ and $u,v\in C$, it holds true that $F(\lambda u+(1-\lambda)v)\geq \lambda F(u)+(1-\lambda)F(v)$. In the rest of the paper, we let $\P$ denote the set of all probability distributions on $\X$; this can  of course be seen as a convex subset of $\mathbb{R}^{|\X|}$ (the probability simplex.) 

\begin{defi}[Uncertainty]\label{def:uncert}
{A
function $U: \P \rightarrow \mathbb{R}$ is an \emph{uncertainty measure} if it is concave  and continuous over $\P$. 
}
\end{defi}

We postpone a full justification of the above definition to Section \ref{sec:roleconcave}. For the time being, we can explain intuitively the role  of concavity   as follows. Suppose the secret is generated  according to either a distribution $p$ or   to another distribution $q$, the choice depending on a   coin toss, with head's probability $\lambda$. The coin toss introduces \emph{extra randomness} in the generation process. Therefore, the overall uncertainty of the adversary about the secret, $U\big(\lambda\cdot p+(1-\lambda)\cdot q\big)$, should   be \emph{no less} than the average uncertainty of the two original generation processes considered separately, that is $\lambda U\big( p\big)+(1-\lambda)U\big( q\big)$. As a matter of fact, most uncertainty measures in \qif\ do satisfy concavity. Continuity is a technical requirement that comes into play only in Theorem~\ref{th:due}\footnote{In fact, concavity does imply continuity except possibly  on the frontier of $\P$.}.

\begin{exa}{
The following entropy functions, and variations thereof, are often considered in the quantitative security literature as measures of the difficulty or effort necessary for a passive adversary to identify a secret $X$, where $X$ is a random variable over $\X$ distributed according to a known distribution $p(\cdot)$. All of these functions are easily proven to be uncertainty    measures in our sense:
\begin{itemize}
\item \emph{Shannon entropy}: $H(p)\defin -\sum_{x\in\X}{p(x)\log p(x) }$, with $0\log 0=0$ and $\log$ in base $2$;
\item \emph{Error probability entropy}: $E(p)\defin 1-\max_{x\in\X}{p(x)}$;
\item \emph{Guessing entropy}: $G(p)\defin \sum_{i=1}^{n-1}{i\cdot p(x_i)}$ with $p(x_1)\geq p(x_2)\geq\ldots\geq p(x_n)$.
\end{itemize}
}
\end{exa}

\begin{exa}\label{ex:variance}
For a somewhat different example of uncertainty, suppose that $\X\subseteq \mathbb{R}$. Then each probability distribution over $\X$ corresponds to a real valued r.v., and we can set $U(X)\defin \var(X)$, where $\var(X)\defin E[(X-\mu)^2]$, with $\mu=E[X]$, is the familiar variance. This makes intuitive  sense, as the higher the variability of $X$, the higher the uncertainty about its value. Let us check that $\var(X)$ is concave and continuous.

Indeed, continuity follows immediately from the definition. Concerning concavity, first note that, for any real $z$, $E[(X-z)^2]=\var(X)+(\mu-z)^2$ (easily checked by writing $(X-z)^2$ as $((X-\mu)+(\mu-z))^2$, then expanding the square and then applying linearity of expectation.) This implies that $E[(X-z)^2] \geq \var(X)$. Now, let $p,q$ be any two distributions on $\X$, let $\lambda\in[0,1]$ and $r=\lambda\cdot p +(1-\lambda)\cdot q$. Denoting by $\mu_u$ and $\var_u$, respectively, the expectation and variance of $X$ taken according to a distribution $u$, we have the following.
\begin{eqnarray*}
\var_r(X) & = & E_r[(X-\mu_r)^2] \\
       & =  & \lambda E_p[(X-\mu_r)^2] + (1-\lambda) E_q[(X-\mu_r)^2]\\
       & \geq & \lambda \var_p(X) +  (1-\lambda) \var_q(X)\,.
\end{eqnarray*}
\end{exa}

We note that the min-entropy function, $H_\infty(p)=-\log \max_x p(x)$, is neither concave nor convex, so it does not fit in the present framework. However, one can at least indirectly express min-entropy via the error probability entropy $E(\cdot)$: $H_\infty(p) = -\log(1-E(p))$.

A  \emph{strategy} is a partial function $\sigma: \Y^* \rightarrow Act$ such
that $\dom(\sigma)$ is non-empty and prefix-closed\footnote{A set
  $B\subseteq \Y^*$ is prefix-closed if whenever $\sigma\in B$ and
  $\sigma'$ is a prefix of $\sigma$ then $\sigma'\in B$.}.    
A strategy is  \emph{finite} if $\dom(\sigma)$ is finite.
The \emph{length} of a finite strategy is defined as $\max\,\{|\xi|: \xi\in \dom(\sigma)\}+1$. For each $n\geq
0$ we will let   $y^n, w^n, z^n,\ldots$   range over sequences in $\Y^n$;   given $y^n=(y_1,\ldots,y_n)$ and $0\leq j\leq n$,   we will let $y^j$    denote the first $j$ components of $y^n$, $(y_1,\ldots,y_j)$.
Given a   strategy $\sigma$ and an integer $n\geq 0$, the  \emph{truncation} of $\sigma$ at level $n$,
denoted as $\sigma \backslash n$, is the finite strategy   $\sigma_{|\cup_{0\leq i\leq n}\Y^i}$.
\begin{wrapfigure}[9]{r}{0.35\textwidth}
\centering
\begin{tikzpicture}[-,>=stealth',shorten >=1pt,auto,node distance=2cm, thick, edge from parent path={(\tikzparentnode) -- (\tikzchildnode)}]
  \tikzstyle{level 1}=[level distance=2.2cm, sibling distance=1.75cm]
  \tikzstyle{every state}=[rectangle,rounded corners, inner sep=2pt,minimum size=.4cm]
\node[state] {{\scriptsize $a$}}
 child
  { node [state]    {{\scriptsize $b$}}
  edge from parent
  node [left=-1pt] {\scriptsize $y$}
  {}}
;
\end{tikzpicture}
\hspace{10pt}
\begin{tikzpicture}[-,>=stealth',shorten >=1pt,
auto,node distance=2cm, thick, edge from parent path={(\tikzparentnode) -- (\tikzchildnode)}]
  \tikzstyle{level 1}=[level distance=1.1cm, sibling distance=1.75cm]
  \tikzstyle{every state}=[rectangle,rounded corners, inner sep=2pt,minimum size=.4cm]
\node[state] {{\scriptsize $a$}}
 child
  { node [state]    {{\scriptsize $b$}}
  child
  { node [state]    {{\scriptsize $d$}}
  edge from parent
  node [left=-2pt] {\scriptsize $y'$}
  {}}
  edge from parent
  node [above] {\scriptsize $y\ \ $}}
  child
  { node [state]    {{\scriptsize $c$}}
  edge from parent
  node [above] {\scriptsize $\ \ y'$}
  {}}
;
\end{tikzpicture}
\vspace{-3pt}
\caption{Two strategy trees.}\label{fig:strategies}
\end{wrapfigure}
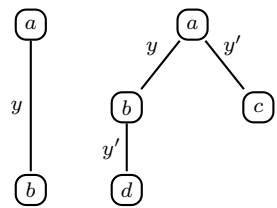
\vspace*{5pt}%
A finite  strategy of length $l$ is \emph{complete} if $\dom(\sigma)=\cup_{0\leq i\leq l-1}\Y^i$.
A strategy $\sigma$ is  \emph{non-adaptive} if whenever $y^n$ and $w^n$ are two
sequences of the same length then $\sigma(y^n)=\sigma(w^n)$ (that is, the decision of which action to play next only depends on the number of past actions); note that finite non-adaptive strategies are necessarily complete.
We note  that strategies can be described as trees, with nodes labelled
by actions and arc labelled by observations, in the obvious way. Any non-adaptive strategy also enjoys a simpler representation as a finite or infinite list  of actions: we write $\sigma=[a_1,\ldots,a_i,\ldots]$ if $\sigma(y^{i-1})=a_i$, for $i=1,2,\ldots$.

\begin{exa}
{
Strategies  $\sigma= [\varepsilon\mapsto a, y\mapsto b]$ and
$\sigma'=[\varepsilon\mapsto a, y\mapsto b, y'\mapsto c, yy'\mapsto d]$ can be represented as in Figure  \ref{fig:strategies}. Note that the tree's height is one less than the strategy's length.}
\end{exa}


\begin{rem}\label{rem:strat}{
It is worthwhile to comment on two possible objections to our strategy model. First,  we do not consider \emph{mixed} strategies, that is strategies where the next action is chosen probabilistically, rather than deterministically like in our case. It is true that mixed strategies play a key role in Game Theory:   equilibria typically arise  in the form of   profiles of mixed  strategies, as in many games any pure (deterministic) strategy could be easily beaten by an opponent. We believe, however, that mixed strategies are irrelevant in the present context:  there is only one player here (the adversary), and consequently no meaningful notion of equilibrium. In particular, there is no such role as a defender whose moves depend on the adversary's ones. In game-theoretical terms, the adversary is playing   against the Nature, represented by the mechanism.

Another possible objection is that in our model the next action depends solely on the sequence of past observations, rather than on the whole history comprising also the past actions played by the adversary. This limitation can be easily overcome by considering a modified action-based mechanism, where actions are part of the observation: that is, in the new mechanism, the set of observations is $Act\times \Y$, and one poses $p_{a}((b,y)|x)\defin p_a(y|x)$ if $a=b$ and  $p_{a}((b,y)|x)\defin 0$ otherwise, where $p_a(\cdot|\cdot)$ is the $a$-stochastic matrix of the original mechanism. This way, strategies for the new mechanism automatically take into account the whole history.
}
\end{rem}

\subsection{A probability space}
Informally, we consider  an adversary who repeatedly queries a mechanism, according to a predefined \emph{finite} strategy. At some point,  the strategy will terminate, and the adversary   will have collected a sequence of observations $y^n=(y_1,\ldots,y_n)$.
Note that both the length $n$ and the probability of the individual observations $y_i$, hence of the whole $y^n$,  will in general depend both on $X$ and on the strategy played by the adversary. In other words, the   distribution  $p(\cdot)$ of $X$ \emph{and} the strategy $\sigma$ together induce a probability distribution  on a subset of all observation sequences: the ones that may arise as a result of a complete interaction with the mechanism, according to the played strategy. 

Formally, let $p(\cdot)$ be any given probability distribution over $\X$, which we will often refer to as the \emph{prior}.
For each   finite strategy $\sigma$, we define a joint  probability
distribution $p_\sigma(\cdot)$ on $\X\times \Y^*$, depending on $\sigma$ and on  $p(\cdot)$, as follows.
We let $p_\sigma(x,\varepsilon)\defin 0$ and, for each $j\geq 0$:
\begin{eqnarray}\label{eq:probspace}
p_\sigma(x,y_1,\ldots,y_j, y_{j+1}) & \defin &
\left\{\begin{array}{ll} p(x)\cdot p_{a_1}(y_1|x)\cdot \cdots \cdot p_{a_{j}}(y_j|x) p_{a_{j+1}}(y_{j+1}|x) \\
\ \qquad\qquad\qquad\qquad\qquad\ \mbox{ if }y^j\in \dom(\sigma), y^jy_{j+1}\notin\dom(\sigma)\\[10pt]
\\
0  \qquad\qquad\qquad\qquad\qquad\ \mbox{otherwise}
\end{array}\right.
\end{eqnarray}

\ifmai
\begin{eqnarray*}
p_\sigma(x,y_1,\ldots,y_j, y_{j+1}) & \defin &
\left\{\begin{array}{ll} p(x)\cdot p_{a_1}(y_1|x)\cdot \cdots \cdot p_{a_{j}}(y_j|x) p_{a_{j+1}}(y_{j+1}|x) & \mbox{ if }y^j\in \dom(\sigma),\\
\ & \phantom{ if} \ y^jy_{j+1}\notin\dom(\sigma)\vspace{-2pt}\\
0 & \mbox{ otherwise}
\end{array}\right.
\end{eqnarray*}
\fi
\noindent where in the first case $a_i = \sigma(y^{i-1})$ for $i =1,\ldots,j+1$. In case $\sigma=[a]$, a single action strategy, we will often abbreviate $p_{[a]}(\cdot)$ as $p_a(\cdot)$. Note that the support of $p_\sigma(\cdot)$ is finite, in particular $\supp(p_\sigma) \subseteq \X\times  \{y^jy: j\geq 0, y^j\in \dom(\sigma), y^jy\notin\dom(\sigma)\}$.

Let $(X,Y)$ be a pair of random variables with outcomes in $\X\times \Y^*$, jointly distributed according to $p_\sigma(\cdot)$:  here $X$ represents the secret and $Y$ represents the sequence of observations obtained upon termination of the strategy.
We shall often use such shortened notations as: $p_\sigma(x|y^n)$ for
$\Pr(X=x|Y=y^n)$, $p_\sigma(y^n)$ for $\Pr(Y=y^n)$,  and so on. Explicit formulas for computing these quantities can be easily derived from the definition of $p_\sigma(\cdot)$ and using Bayes rule.
We will normally keep the dependence of $(X,Y)$ on $p(\cdot)$ and $\sigma$
implicit. When we
want to stress that we are considering   $Y$ according to the distribution  induced by  a specific $\sigma$ (e.g. because different strategies are being considered at the same time), we will write it as $Y_\sigma$.

\vsp
Consider a   prior $p(\cdot)$ and a \emph{finite} strategy $\sigma$, and the corresponding pair of \fra{se vogliamo aggiungere questa notazione non perdiamo spazio. Decidi te. Ho scritto per esteso anche w.l.o.g., tanto ho visto che compariva solo 1 volta in appendice}
   random variables (r.v.)  $(X,Y)$.
We define the following quantities,
expressing  average  uncertainty, conditional uncertainty  and information gain about $X$, that may result from interaction according to strategy $\sigma$ (by convention, below we let   $y^n$ range over sequences with $p_\sigma(y^n)>0$):
\renewcommand{\arraystretch}{0.9}
\begin{eqnarray}
U(X) &\defin & U\big(p \big)\nonumber\\
U_\sigma(X|Y) &\defin &\sum_{y^n} p_\sigma(y^n)U\big(p_\sigma(\cdot|y^n)\big)\label{eq:conditionaluncert}\\ 
I_\sigma(X;Y)&\defin & U(X)-U_\sigma(X|Y)\,.\nonumber
\end{eqnarray}

Again, we may drop the subscript $\sigma$ from $U_\sigma$ and $I_\sigma$ if the strategy $\sigma$ is clear from the context.  Note that, in the case of Shannon entropy, $I_\sigma (X;Y)$ coincides with the familiar mutual information, traditionally measured in bits. In the case of error entropy, $I_\sigma (X;Y)$ is what is called \emph{additive   leakage}  in  e.g. \cite{BCP09}  and \emph{advantage} in the cryptographic literature, see e.g. \cite{entropicSec} and references therein. 

In the rest of the paper, unless otherwise stated, we let $U(\cdot)$ be an arbitrary uncertainty function. The following fact about $I_\sigma(X;Y)$ follows from $U(\cdot)$'s concavity and Jensen's inequality, plus routine calculations on probability distributions (see Appendix).

\begin{lem}\label{lemma:uno}   $I_\sigma(X;Y)\geq 0$. Moreover $I_\sigma(X;Y)= 0$ if $X$ and $Y$ are independent.
\end{lem}


Given the   definitions in \eqref{eq:conditionaluncert},    adaptive \qif\ can   be defined quite simply.

\begin{defi}[\qif\ under adaptive adversaries]\label{def:qifAdaptive}
Let $\St$ be a mechanism and $p(\cdot)$ be a prior over $\X$.
\begin{enumerate}
\item For a \emph{finite} strategy $\sigma$, let  $I_\sigma(\St,p) \defin I_\sigma(X;Y)$.
\item For an \emph{infinite} strategy  $\sigma$,
  let $I_\sigma(\St,p) \defin \lim_{l\rightarrow\infty} I_{\sigma\backslash l}(\St,p)$.
\item (\emph{Maximum IF} under  $p(\cdot)$) $I_\star(\St,p)\defin  \sup_\sigma I_\sigma(\St,p)$.
\end{enumerate}
\end{defi}

Note that $l'\geq l$ implies $I_{\sigma\backslash l'}(\St,p)\geq I_{\sigma\backslash l}(\St,p)$, hence the limit in (2) always exists.
Taking the distribution that achieves the maximum leakage, we can define an analog of channel capacity.

\begin{defi}[Adaptive secrecy capacity]
 $C(\St) \defin \sup_{p \in \P} I_\star(\St,p)$.
\end{defi}

\subsection{Attack Trees}
It is sometimes useful to work with a pictorial representation of the adversary's attack steps, under a given strategy and prior. This can take  the form of a tree, where each node   represents an adversary's \emph{belief} about the secret, that is, a probability distribution over $\X$. The tree describes the possible evolutions of the  belief, depending on the strategy and on the observations.  We formally introduce such a representation below: it will be extensively used in the examples. Note that \emph{attack} trees
are different from  \emph{strategy} trees. 

A \emph{history} is a sequence $h\in (Act\times\Y)^*$. Let
$h=(a_1,y_1,\ldots,a_n,y_n)$ be such a history.
Given a prior $p(\cdot)$, we define the \emph{update of $p(\cdot)$ after $h$}, denoted by
$p^h(\cdot)$,  as the distribution on $\X$ defined by
\begin{eqnarray}\label{eq:update}
p^h(x) &\defin & p_{\sigma_h}(x | y^n)
\end{eqnarray}
\noindent where $\sigma_h=[a_1,\ldots,a_n]$,  provided $p_{\sigma_h}(y^n)>0$;
otherwise $p^h(\cdot)$ is undefined.

The \emph{attack tree} induced by a   strategy $\sigma$ and a prior
$p(\cdot) $ is  a tree with nodes labelled by probability distributions over $\X$ and arcs labelled with pairs $(y,\lambda)$ of an observation $y$ and a probability value $\lambda$. This tree is obtained from the strategy tree of $\sigma$ as follows.  First, note that, in a strategy tree, each node can be identified with the unique history from the root leading to it. Given the strategy tree for $\sigma$: (a) for each $y\in \Y$ and each node missing an outgoing $y$-labelled arc,  attach   a new     $y$-labelled arc  leading to a new node;  (b) label each node of the resulting tree by $p^h(\cdot)$, where $h$ is the history  identifying the node,
if $p^h(\cdot)$ is defined,  otherwise remove the node and its descendants, as well as the incoming arc; (c) label each arc from a node $h$ to a child, represented by a history $h\cdot a\cdot y$,  in the resulting tree with  $\lambda=p_a^{h}(y)$ - to be parsed as $(p^h)_{[a]}(y)$. This  is  the probability of observing $y$ under a prior $p^h(\cdot)$ when submitting action $a$.

The concept of attack trees is demonstrated by a few examples in the next section.
Here, we just note the following  easy to check facts. For each leaf $h$ of the attack tree: (i) the leaf's label is $p^h(\cdot)=p_\sigma(\cdot|y^n)$, where $y^n$ is the sequence of observations in $h$;  (ii) if we let $\pi_h $ be the product of the probabilities on the edges from the root to the leaf, then $\pi_h=p_\sigma(y^n)$. Moreover, (iii) each $y^n$ s.t. $p_\sigma(y^n)>0$ is found in the tree. As a consequence, for a \emph{finite} strategy, taking \eqref{eq:conditionaluncert} into account, the uncertainty of $X$ given $Y$ can be computed from the attack tree as:
\begin{eqnarray}\label{eq:uncertree}
U_\sigma(X\,|\,Y) & = & \sum_{h  \ \mathrm{is\  a\  leaf}} \pi_h U(p^h)\,.
\end{eqnarray}
\section{Examples}\label{sec:examples}
We present a few instances of the framework introduced in the previous section. We emphasize that these examples  are quite simple and only serve  to illustrate our main definitions. In the rest of the paper, we shall use the following notation:  we let  $\unif\{x_1,\ldots,x_k\}$ denote the uniform distribution   on $\{x_1,\ldots,x_k\}$.

\begin{figure}[t]
\begin{minipage}[t]{0.50\textwidth}
\centering
\renewcommand{\arraystretch}{0.85}
\begin{tabular}{|>{\centering\arraybackslash}p{0.3cm}|>{\centering\arraybackslash}p{0.5cm}|
>{\centering\arraybackslash}p{0.5cm}|>{\centering\arraybackslash}p{0.6cm}|c|}
\hline
\scriptsize\textbf{id}   &\scriptsize\textbf{ZIP}     &\scriptsize\textbf{Age}   &\scriptsize\textbf{Date}   &\scriptsize\textbf{Disease}\\ \hline
\scriptsize$1$  &\scriptsize$z_1$   &\scriptsize$65$  &\scriptsize$d_2$  &\scriptsize Heart disease\\ \hline
\scriptsize$2$  &\scriptsize$z_1$   &\scriptsize$65$  &\scriptsize$d_2$  &\scriptsize Flu\\ \hline
\scriptsize$3$  &\scriptsize$z_1$   &\scriptsize$67$  &\scriptsize$d_2$  &\scriptsize Short breath\\ \hline
\scriptsize$4$  &\scriptsize$z_1$   &\scriptsize$68$  &\scriptsize$d_1$  &\scriptsize Obesity\\ \hline
\scriptsize$5$  &\scriptsize$z_1$   &\scriptsize$68$  &\scriptsize$d_1$  &\scriptsize Heart disease\\ \hline
\scriptsize$6$  &\scriptsize$z_3$   &\scriptsize$66$  &\scriptsize$d_2$  &\scriptsize Heart disease\\ \hline
\scriptsize$7$  &\scriptsize$z_3$   &\scriptsize$67$  &\scriptsize$d_2$  &\scriptsize Obesity\\ \hline
\scriptsize$8$  &\scriptsize$z_3$   &\scriptsize$31$  &\scriptsize$d_2$  &\scriptsize Short breath\\ \hline
\scriptsize$9$  &\scriptsize$z_2$   &\scriptsize$30$  &\scriptsize$d_3$  &\scriptsize Heart disease\\ \hline
\scriptsize$10$  &\scriptsize$z_2$   &\scriptsize$31$  &\scriptsize$d_3$  &\scriptsize Obesity\\ \hline
\end{tabular}
\caption{\small Medical {\sc db} of Example  \ref{exDBdet}.}
\label{database}
\end{minipage}
\hspace{15pt}
\begin{minipage}[t]{0.45\textwidth}
\centering
\tikzstyle{bag} = [text centered,circle, minimum width=3pt,fill, inner sep=0pt]
\tikzstyle{Node}=[draw, minimum width=1.5pt,fill, inner sep=1.5pt]
\begin{tikzpicture}[-,>=stealth',shorten >=1pt,auto,node distance=2cm, thick, edge from parent path={(\tikzparentnode) -- (\tikzchildnode)}, scale=0.7]
  \tikzstyle{level 1}=[level distance=1.2cm, sibling distance=1.75cm]
  \tikzstyle{every state}=[rectangle,rounded corners, inner sep=2pt,minimum size=.45cm]

\node[state] {{\scriptsize ZIP}}
 child
  { node [state]    {{\scriptsize Date}}
   edge from parent
                node[above=4pt, left=-7pt] {\scriptsize $z_1$}
   {}
   }
 child
  { node [state]    {{\scriptsize Date}}
  edge from parent
                node[left=-2pt] {\scriptsize $z_2$}
   {}
  }
  child
  { node [state]    {{\scriptsize Age}}
  edge from parent
                node[above=4pt, right=-7pt] {\scriptsize $z_3$}
   {}
  }
;
\end{tikzpicture}
\caption{\small Strategy tree of Example ~\ref{exDBdet}.}
\label{database3}
\end{minipage}
\end{figure}

\begin{exa}[medical {\sc db}]{
An attacker gets hold of the table shown in Figu\label{exDBdet}re  \ref{database}, which
represents  a fragment of a hospital's   database. Each row of the table  contains:  a numerical id followed by  the ZIP code, age,   discharge date and disease of an individual that has
been recently hospitalized. The table does not contain personal identifiable information. The attacker gets to know that a certain target individual, John Doe (JD),   has been recently hospitalized. However, the attacker is ignorant of the corresponding id in the table and any information about JD, apart from his name. The attacker's task  is to identify JD, i.e. to  find  JD's id in the table, thus learning his disease. The attacker is in a position to ask a   source, perhaps the hospital {\sc db},     queries  concerning non sensitive information (ZIP code, age and  discharge date) of any individual, including JD, and compare the answers with the table's entries.\footnote{That this is  unsafe is of course well-known  from   database security:  the present example only serves the purpose of illustration.}

This situation can be modeled quite simply as an action-based mechanism $\St$, as follows.
 We pose: $Act=\{\mathrm{ZIP,Age,Date}\}$; $\X=\{1,\ldots,10\}$, the set of possible  id's, and $\Y=\Y_{\mathrm{ZIP}}\cup \Y_{\mathrm{Age}}\cup\Y_{\mathrm{Date}}$, where $\Y_{\mathrm{ZIP}}=\{z_1,z_2,z_3\}$,  $\Y_{\mathrm{Age}}=\{30,31,65,66,67,68 \}$ and $\Y_{\mathrm{Date}}=\{d_1,d_2,d_3\}$. The   conditional probability matrices reflect the behaviour of the   source when queried about ZIP code, age and discharge date of an individual. We assume that the source is truthful, hence   answers will match   the entries    of the table. For example, $p_{\mathrm{Age}}(y|1)=1$ if $y=65$ and $0$ otherwise;
 $p_{\mathrm{ZIP}}(y|2)=1$ if $y=z_1$, $0$ otherwise; and so on.  Note that this defines a \emph{deterministic} mechanism.   Finally, since the attacker has no clues about JD's id, we set the prior to be the uniform distribution on $\X$, $p(\cdot)=\unif\{1,\ldots,10\}$.

Assume now that, possibly to protect privacy of individuals,  the number of queries to the   source  about any individual is limited to two. Figure  \ref{database3} displays a possible attacker's strategy $\sigma$, of length 2.  Figure  \ref{fig:attdet} displays the corresponding attack tree, under the given   prior. Note that the given strategy is not in any sense optimal. Assume we set $U(\cdot)=H(\cdot)$,   Shannon entropy, as a measure of   uncertainty.  Using \eqref{eq:uncertree}, we can compute $I_\sigma(\St,p)=H(X)-H(X|Y)=\log{10}-\frac{3}{10}\log{3}-\frac{2}{5}\approx 2.45$ bits. With $U(\cdot)=E(\cdot)$,
the error entropy, we have $I_\sigma(\St,p)= E(X)-E(X|Y) = 0.5.$
}
\end{exa}

\begin{figure}[t]
\begin{minipage}[t]{0.35\textwidth}
\centering
\begin{tikzpicture}[-,>=stealth',shorten >=1pt,auto,node distance=2cm, thick, edge from parent path={(\tikzparentnode) -- (\tikzchildnode)},scale=0.7]
  \tikzstyle{every state}=[rectangle,rounded corners, inner sep=2pt,minimum size=.45cm]
  \tikzstyle{level 1}=[level distance=1.3cm, sibling distance=3cm]
\tikzstyle{level 2}=[level distance=2.3cm, sibling distance=1.3cm]

\node[state]{{\scriptsize $\unif \{1,\ldots,10\}$}}
 child
  {
  node [state]{{\scriptsize $\unif \{1,\ldots,5\}$ }}
     child { node [state, left] {{\scriptsize $\unif \{1,2,3\}$ }}
             edge from parent
                node[above=-1pt] {\scriptsize $d_2\ \ $}
                 node[below=5pt, right=2pt]  {\scriptsize  $\!\!\!\!\!\frac{3}{5}$} }
      child { node [state, right] {{\scriptsize $\unif \{4,5\}$ }}
             edge from parent
                node[above=-1pt] {\scriptsize $\ \ d_1$}
                 node[below=5pt,left=-2pt]  {\scriptsize  $\frac{2}{5}$} }
     edge from parent
                node[above] {\scriptsize $z_1$}
                node[below right]  {\scriptsize $\!\!\!\!\!\!\!\!\frac{1}{2}$}
   }
 child
  { node [state]    {{\scriptsize $\unif \{9,10\}$}}
     child { node [state] {{\scriptsize $\unif \{9,10\}$ }}
             edge from parent
                node[below =-2pt, left =-2pt] {\scriptsize $d_3$}
                 node[below  right]  {\scriptsize  $1$} }
    edge from parent
                node[above= 2pt, left=-1pt] {\scriptsize $z_2$}
                node[below=-1pt, right=-1pt]  {\scriptsize  $\frac{1}{5}$}
     {} }
  child
  { node (P) [state]    {{\scriptsize $\unif \{6,7,8\}$}}
    child { node (Q) [state] {{\scriptsize $\unif \{8\}$ }}
             edge from parent
                node[above=4pt, left=-5pt] {\tiny $31$}
                node[below=5pt,right=-4pt]  {\scriptsize  $\frac{1}{3}$} }
    child { node (Q1) [state] {{\scriptsize $\unif \{6\}$ }}
             edge from parent
                node[above =3pt,left=-3pt] {\tiny $66$}
                 node[above=-5pt, right=-3pt]  {\scriptsize  $\frac{1}{3}$} }
    child { node (Q2) [state] {{\scriptsize $\unif \{7\}$ }}
             edge from parent
                node[above=4pt, right=-5pt] {\tiny $67$}
                 node[below=5pt,left=-5pt]  {\scriptsize  $\frac{1}{3}$} }
  edge from parent
                node[above] {\scriptsize $z_3$}
                node[below]  {\scriptsize $\frac{3}{10}$}
    {}  }
;
\end{tikzpicture}
\caption{\small The attack tree for Example  \ref{exDBdet}.}
\label{fig:attdet}
\end{minipage}
\hspace{40pt}
\begin{minipage}[t]{0.55\textwidth}
\centering
\begin{tikzpicture}[-,>=stealth',shorten >=1pt,auto,node distance=2cm, thick, edge from parent path={(\tikzparentnode) -- (\tikzchildnode)},scale=0.7]
  \tikzstyle{every state}=[rectangle,rounded corners, inner sep=2pt,minimum size=.45cm]
  \tikzstyle{level 1}=[level distance=1.3cm, sibling distance=3.7cm]
\tikzstyle{level 2}=[level distance=2.3cm, sibling distance=1.45cm]

\node[state]{{\scriptsize $\unif \{1,\ldots,10\}$}}
 child
  {
  node [state]{{\scriptsize $\unif \{1,\ldots,5\}$ }}
     child { node [state, left=0.5pt] {{\scriptsize $\unif \{1,2,3\}$ }}
             edge from parent
                node[above=-1pt] {\scriptsize $d_2\ \ $}
                 node[below=5pt, right=2pt]  {\scriptsize  $\!\!\!\!\!\frac{3}{5}$} }
      child { node [state, right=0.5pt] {{\scriptsize $\unif \{4,5\}$ }}
             edge from parent
                node[above=-1pt] {\scriptsize $\ \ d_1$}
                 node[below=5pt, left=-2pt]  {\scriptsize  $\frac{2}{5}$} }
     edge from parent
                node[above] {\scriptsize $z_1$}
                node[below right]  {\scriptsize $\!\!\!\!\!\!\!\!\frac{1}{2}$}
   }
 child
  { node [state]    {{\scriptsize $\unif \{9,10\}$}}
     child { node [state] {{\scriptsize $\unif \{9,10\}$ }}
             edge from parent
                node[below =-2pt, left =-2pt] {\scriptsize $d_3$}
                 node[below  right]  {\scriptsize  $1$} }
    edge from parent
                node[above= 2pt, left=-1pt] {\scriptsize $z_2$}
                node[below=-1pt, right=-1pt]  {\scriptsize  $\frac{1}{5}$}
     {} }
  child
  { node (P) [state]    {{\scriptsize $\unif \{6,7,8\}$}}
    child { node (Q) [state] {{\scriptsize $\unif \{8\}$ }}
             edge from parent
                node[above=-2pt,sloped] {\tiny $30,31,32\ \ $}
                node[below=5pt, right=-5pt]  {\scriptsize  $\frac{1}{9}$} }
    child { node (Q1) [state] {{\scriptsize $\unif \{6\}$ }}
             edge from parent
                node[above =-3pt,sloped] {\tiny $65$}
                 node[below=5pt, right=-9pt]  {\scriptsize  $\ \ \ \frac{1}{9}$} }
    child { node (Q2) [state] {{\scriptsize $\unif \{6,7\}$ }}
             edge from parent
                node[above=-3pt,sloped] {\tiny $66,67$}
                 node[below=5pt, left=-4pt]  {\scriptsize  $\frac{2}{9}$} }
    child { node (R)[state] {{\scriptsize $\unif \{7\}$ }}
             edge from parent
                node[above =-2pt,sloped] {\tiny 68 }
                 node[below=5pt, left=-5pt]  {\scriptsize  $\frac{1}{9}$} }
  edge from parent
                node[above] {\scriptsize $z_3$}
                node[below]  {\scriptsize $\frac{3}{10}$}
    {}  }
;
\path (P) -- coordinate[above] (PQ) (Q);
\path (P) -- coordinate[above] (PR) (R);
\end{tikzpicture}
\caption{\small The attack tree for Example  \ref{exDB}. Leaves   with the same label and their incoming arcs have been coalesced.}
\label{database2}
\end{minipage}
\end{figure}


\begin{exa}[noisy version]\label{exDB}{
We consider a version of the previous mechanism where the pu\-blic source  queried by the attacker is not entirely truthful. In particular,   for  security reasons, whenever queried about age of an individual, the source adds a uniformly distributed offset $r\in\{-1,0,+1\}$ to the real answer. The only difference from the previous example is that the conditional probability matrix $p_{\mathrm{Age}}(\cdot|\cdot)$  is not deterministic anymore. For example, for $x=1$, we have
\begin{eqnarray*}
 p_{\mathrm{Age}}(y|1) & = & \left\{\begin{array}{ll}\frac 1 3 & \mbox{ if }y\in\{64,65,66\}\\
                                                            0 & \mbox{otherwise }
                                    \end{array}\right.
 \end{eqnarray*}
(also note that we have to insert $29, 32, 64 \mbox{ and } 69$ as possible mechanism's observations into $\Y_{\mathrm{Age}}$.)
Figure  \ref{database2} shows the attack tree induced by the strategy $\sigma$ of Figure  \ref{database3} and the uniform prior in this case. If $U(\cdot)=H(\cdot)$  we obtain $I_\sigma(\St,p)=\log{10}-\frac{3}{10}\log{3}-\frac{8}{15}\approx 2.31$ bits; if  $U(\cdot)=E(\cdot)$, instead, $I_\sigma(\St,p)= \frac{13}{30} \approx 0.43.$}
\end{exa}


\begin{exa}[cryptographic devices]\label{ex:crypt}{
We can abstractly model a cryptographic device as a function $f$ taking   pairs of a key and a message into observations, thus, $f:\K \times \M \rightarrow \Y$.  Assume the attacker can  choose the message $m\in \M$ fed to the device, while the private key $k$ is fixed and unknown to him. This clearly yields an action-based mechanism $\St$ where $\X=\K$, $Act=\M$ and $\Y$ are the observations. If we assume the observations  noiseless, then the conditional probability matrices are defined by
\begin{eqnarray*}
p_m(y|k) & = & 1 \ \ \mbox{ iff }\ \  f(k,m)=y   \,.
\end{eqnarray*}
We obtain therefore a deterministic mechanism. This is the way, for example, modular exponentiation is modeled in \cite{KB07}. More realistically, the observations will be noisy, due e.g. to the presence of  ``algorithmic noise''. For example, assume $\Y\subseteq \mathbb{N}$ is the set of possible Hamming weights of the ciphertexts (this is related to power analysis attacks, see e.g. \cite{KSWH00}.)
Then we may set
\begin{eqnarray*}
p_m(y|k) & = & \Pr(f(k,m)+N =y)
\end{eqnarray*}
\noindent where $N$ is a random variable modelling noise. For example, in the model  of {\sc des} S-Boxes considered in \cite{BP12}, $\K=\M=\{0,1\}^6$, while $\Y=\{0,1,2,\ldots\}$ is the set of    observations: the (noisy) Hamming weight  of the outputs of the target S-Box. In this case, $N$  is taken to be   the  cumulative weight of the seven S-Boxes other than the target one. It is   sensible to assume this   noise to be   binomially distributed: $N\sim B(m,p)$, with $m=28$ and $p=\frac  1 2$. See \cite{BP12} for details.
}
\end{exa}

\ifmai
\begin{exa}
Let us consider the game Mastermind, a classic two-player game between a codemaker and a codebreaker. The codemaker secretly selects a $n$-digit code, where each digit is between $1$ and $m.$ The codebreaker then tries to guess the code by repeatedly guessing a valid code ($n$ digits belonging to $\{1,2,\ldots,m\}$.) After each guess, the codemaker tells the codebreaker two numbers, $(b, w):$ the number of correct digits in the correct position $(b)$ and the number of digits that are part in the code but not in the correct position $(w.)$ For example, if the code is $1233$ and the codebreaker's guess is $3243,$ then the codemaker's response would be $b=2, w=1,$ since the codebreaker has guessed the $2$ and the second $3$ correctly and in the correct positions, while the first $3$ is in the wrong position.
Consider now the case of a careless codemaker, that each time with a certain probability can make mistakes
and give a wrong answer. Let us denote by $b'$ and $w'$ the reported answers and let $q(\cdot|(b,w))$
represent the probability of obtaining a certain answer, given the real one.
Suppose then that in our case $n=3$ and $m=2.$ We can view it as an
action-based information hiding mechanism $\St$ where $\X=\Act={\{0,1\}}^3,$ $\Y=\{(b',w')\in{\{0,1,2,3\}}^2|b'+w'\leq 3\}
\setminus\{(2,1),(0,1),(0,3)\}$\footnote{It is easy to see that these three couples are not possible values for $(b,w):$
in the case of $(2,1),$ for example, if we have two correct values in the correct position ($b=2$) on a total of three,
it can not be possible that the third one is correct but in a wrong position ($w=1$) because there are not other possible positions it can occupy.} and for each action $a\in\Act$ the correspondent probability matrix
will be given by the product $Q\cdot P_a,$ where $P_a$ is the matrix corresponding to the deterministic case
($P_a((b,w)|k)=1$ if and only if, in the code $a,$ $b$ values are correct and in the correct position and $w$
correct but in a wrong position) and $Q$ contains the probabilities $q((b',w')|(b,w).)$ In other words,
$p((b',w')|k)=q((b',w')|(b,w))p((b,w)|k.)$
\end{exa}
\fi

\section{Comparing Adaptive and Non-adaptive Strategies}\label{sec:compare}

Conceptually, we can broadly  classify mechanisms  into two categories,  depending on the size of the set $Act$.  The first category consists of systems with a huge   - exponential, in the size of any reasonable syntactic description - number of actions. The second category consists of systems with an ``affordable''  number of actions. In the first category, we find, for instance, complex cryptographic hardware, possibly described via boolean circuits or other ``succinct'' notations (cf.  the public key exponentiation algorithms considered in    \cite{KB07}.) In the second category, we find  systems explicitly described by tables, such as databases (Example  \ref{exDBdet} and \ref{exDB}) and  S-Boxes  (Example \ref{ex:crypt}.) It makes sense to assess the difficulty of analysing the security of   mechanisms separately for these two categories.

\subsection{Systems in  Succinct Form}
We argue that the analysis of such systems is in general an intractable problem, even if restricted to simple special instances of the \emph{non-adaptive} case. We consider the problem of deciding if there is a  finite strategy  over a given time horizon   yielding an information flow exceeding a given threshold. This decision problem is of course simpler than the problem of finding an optimal strategy over a finite time horizon: indeed, any algorithm for finding the optimal strategy can also be used to answer the first problem. We give some definitions.

\begin{defi}[Systems in  boolean forms]\label{def:boolean} Let $t,u,v$ be nonnegative integers. We say a mechanism $\St=(\X,\Y,Act,\{M_a:a\in Act\})$ is in $(t,u,v)$-\emph{boolean form} if $\X=\{0,1\}^t$, $Act=\{0,1\}^u$, $\Y=\{0,1\}^v$   and   there is a boolean function $f:\{0,1\}^{t+u}\rightarrow \{0,1\}^v$ such that for each $x\in \X$, $y\in \Y$ and $a\in Act$, $p_a(y|x)=1$ iff $f(x,a)=y$. The \emph{size} of $\St$ is defined as the syntactic size of the smallest boolean formula for $f$.
\end{defi}

It is not difficult to see that the class of boolean forms coincides, up to   suitable encodings, with that of deterministic systems.

\begin{defi}[Adaptive Bounding Problem in succinct form, \textsc{abps}]\label{def:ADP} Given a mechanism $\St$  in a $(t,u,v)$-boolean form, represented by a boolean expression, a prior distribution $p(\cdot)$, $l\geq 1$ and $T\geq 0$,
decide if there is a strategy $\sigma$ of length $\leq l$ such that
$I_\sigma(\St;p)> T$.
\end{defi}

In the following theorem, we shall assume, for simplicity, \label{marker}
the following reasonable properties of $U(\cdot)$: if  $p(\cdot)$ concentrates all the probability mass on a single element, and  $q(\cdot)$ is the uniform distribution,  then $0=U(p)<U(q)$. A slight modification of the argument   also works without this assumption. 
The theorem says that even length 1 (hence non-adaptive) strategies are difficult to assess.

\begin{thm}\label{th:quattro} Assume $U(\cdot)$ satisfies the above stated property. Then the \textsc{abps} is \textsc{np}-hard, even if fixing $t=v=l=1$,     and $T=0$.
\end{thm}
\begin{proof}
We reduce from the satisfiability problem for boolean formulae. Let $\phi(z_1,\ldots,z_u)=\phi(\tld z)$
be an arbitrary boolean formula  with $u$ free boolean variables $z_1,\ldots,z_u$.
We show how to build in polynomial time out of $\phi(\tld z)$ a mechanism $\St$ in $(1,u,1)$-boolean form,
and a prior $p(\cdot)$, with the following property: there is a  length 1 strategy $\sigma$
s.t. $I_\sigma(\St,p)>0$ iff $\phi(\tld z)$ is satisfiable.
Take $\X=\Y=\{0,1\}$ and $Act=\{0,1\}^u$. Let  the mechanism $\St$ be defined by
the boolean function $f(x,z_1,\ldots,z_u)=x \wedge \phi(\tld z)$. Let $p(\cdot)$
be the uniform prior on $\X=\{0,1\}$.  Now, if there is an action
 $\tld b = (b_1,\ldots,b_u)\in Act$ such that $\phi(\tld b)=1$ ($\phi(\tld z)$ is satisfiable)
  then clearly we will have that $Y=X\wedge \phi(\tld b)$ is logically
  equivalent to $X$, hence $U(X|Y)=0$. Consequently,   setting $\sigma=[\varepsilon\mapsto \tld b]$, we will have
  that  $I_\sigma(\St,p)= U(X)-U(X|Y) >0$.
    On the other hand, if $\phi(\tld z)$ is not satisfiable, then for any $\tld b\in Act$
     we will have that $Y=X\wedge \phi(\tld b)$ is logically equivalent to $0$, hence $U(X|Y)= U(X)$.
     Consequently,
for any $\sigma=[\varepsilon\mapsto \tld b]$, we will have $I_\sigma(\St,p)= U(X)-U(X|Y) =0$.
\end{proof}

We should stress again that the above result concerns the difficulty of \emph{analyzing} succinct mechanisms under the simplest possible form of attacker; by no means it entails that the adaptive and non adaptive attackers are equally effective. The following example should clarify that between the two forms of attacker, there can be a huge difference in terms of effectiveness.

\begin{exa}[envelopes]\label{ex:envelopes}
A secret bit $s\in \{b_0,b_1\}$ and  a numbered envelope $e\in\{1,...,N\}$ are drawn according to some distribution $p$. A piece of paper with the value of $s$ written on  it is put into $e$. All other envelopes are filled with a piece of paper revealing the envelope containing the secret, that is $e$. The adversary can choose and open any  of the envelopes and examine its content: that is, actions of this mechanism are envelope numbers.

Assume the envelope is chosen uniformly at random.
Clearly, the obvious adaptive strategy leads the adversary to discover the secret after at most two actions. On the other hand, a non-adaptive, brute-force strategy will lead him to examine one envelope after the other, ignoring the suggestion given in the envelopes opened so far, leading to a strategy of length   $N-1$.
\end{exa}

\subsection{General Systems}
The following results  apply in general,  but they are   particularly significative for  systems with a  moderate
 number of actions. The next theorem essentially says that, up to an expansion factor bounded by $|Act|$, non-adaptive strategies are as efficient as adaptive ones. 
In fact, given any strategy $\sigma$, one can construct a non-adaptive $\sigma'$ that is only moderately larger than $\sigma$ and achieves at least the same leakage, as follows.
In any history induced by $\sigma$, each action can occur at most $l$ times, where $l$ is the length of $\sigma$, and the order in which different actions appear in the history is not relevant as to the final belief that is obtained. For any history of $\sigma$ to be simulated by an history of $\sigma'$, it is therefore enough that the latter offers all the distinct actions offered by $\sigma$, each repeated $l$ times. Note that, for a strategy $\sigma$, the
number of distinct actions that appear in $\sigma$ is $|\range(\sigma)|$.

\begin{thm}\label{th:cinque} For each finite strategy $\sigma$  of length $l$ it is possible to
\fra{Dobbiamo specificare da qualche parte che lavoriamo con sistemi memoryless?}
build a non-adaptive finite strategy $\sigma'$ of length
$|\range(\sigma)|\times  l $, such that $I_{\sigma'}(\St,p) \geq I_\sigma(\St,p)$.
\end{thm}
\begin{proof}
Let $\range(\sigma)=\{a_1,\ldots,a_h\}$ and let $\sigma'$ be any non-adaptive strategy that plays each of
 $a_1,\ldots,a_h$
  for $l$ times, for example,
$\sigma' = [a_1,\ldots,a_h,\ldots,a_1,\ldots,a_h]$ ($l $ times); note
that the
length of $\sigma'$ is $h\times l $, as required.
For any $y^j$ ($j\leq l$), we shall
denote by $\sigma'-y^j$ the non-adaptive strategy of length $h\times l -j$ obtained
by removing from $\sigma'$, seen as a list,   $j$ actions  $b_1,\ldots,b_j$,
where
$b_1=\sigma(\varepsilon),\ldots,b_j=\sigma(y^{j-1})$.

Denote by $Y_\sigma$ and $Y_{\sigma'}$ the  r.v. on $\Y^*$ corresponding to $\sigma$ and $\sigma'$, respectively.
 We will
show
that $U(X|Y_\sigma)\geq U(X|Y_{\sigma'})$.  Take any $x\in \X$ s.t. $p(x)>0$
and $y^j\in \dom(p_\sigma)$. We note that, for any
sequence $y^{hl-j}$, and  for an appropriate interleaving of the two sequences $y^{hl-j}$ and $y^j$
that here we denote  by just  $y^{hl-j},y^j$,
we have that
 \begin{eqnarray}
 p_{\sigma'-y^j}(y^{hl-j}|x)p_{\sigma}(y^j|x)= p_{\sigma'}(y^{hl-j},y^j|x)\,.\label{eq:decomposition}
 \end{eqnarray}

\noindent From \eqref{eq:decomposition},
 it follows that
 \begin{eqnarray}
 p_{\sigma}(y^j|x) & = & \sum_{y^{hl-j}} p_{\sigma'-y^j}(y^{hl-j}|x)p_{\sigma}(y^j|x)  \ \ \ = \ \ \ \sum_{y^{hl-j}} p_{\sigma'}(y^{hl-j},y^j|x)\,.\label{eq:interleave}
 \end{eqnarray}

\noindent Now, for any $x$ and $y^j$ such that $p(x)>0$ and $p_\sigma(y^j)>0$, we have the following.
\begin{eqnarray}
 p_{\sigma}(x|y^j) & = & \frac{ p_{\sigma}(y^j|x)p(x)}{p_\sigma(y^j)}\nonumber\\
 & = & \sum_{y^{hl-j}} p_{\sigma'}(y^{hl-j},y^j|x) \frac{p(x)}{p_\sigma(y^j)}\label{eq:applyInterl}\\
 & = & \sum_{y^{hl-j}} \frac{p_{\sigma'}(x|y^{hl-j},y^j)p_{\sigma'}(y^{hl-j},y^j)}{p(x)}
                        \frac{p(x)}{p_\sigma(y^j)}\nonumber\\
 & = & \sum_{y^{hl-j}}  p_{\sigma'}(x|y^{hl-j},y^j)\frac{p_{\sigma'}(y^{hl-j},y^j)}{p_\sigma(y^j)} \label{eq:convex}
 \end{eqnarray}
\noindent where in the second equality of \eqref{eq:applyInterl} we have applied \eqref{eq:interleave}. It is an easy matter to show that $\sum_{y^{hl-j}}
\frac{p_{\sigma'}(y^{hl-j},y^j)}{p_\sigma(y^j)}=1$ (this is basically a consequence of \eqref{eq:decomposition};
we leave the details to the interested reader.)
Thus \eqref{eq:convex}
shows that $ p_{\sigma}(\cdot|y^j)$
can be expressed as a convex combination of the distributions $p_{\sigma'}(\cdot|y^{hl-j},y^j)$, for
 $y^{hl-j}\in \Y^{hl-j}$.
 Using this fact, the concavity of
$U(\cdot)$ and Jensen's inequality, we arrive at the following.
\begin{eqnarray}
U(p_{\sigma}(\cdot|y^j)) & \geq & \sum_{y^{hl-j}}  U(p_{\sigma'}(\cdot|y^{hl-j},y^j))
                    \frac{p_{\sigma'}(y^{hl-j},y^j)}{p_\sigma(y^j)}\,.\label{eq:convex2}
\end{eqnarray}

\noindent We finally can compute the following lower-bound for $U(X|Y_\sigma)$.
\begin{eqnarray}
U(X|Y_\sigma) & = & \sum_{y^j} p_\sigma(y^j) U(p_{\sigma}(\cdot| y^j))\nonumber\\
& \geq & \sum_{y^j} p_\sigma(y^j) \sum_{y^{hl-j}} \frac{p_{\sigma'}(y^{hl-j},y^j)}{p_\sigma(y^j)}U(p_{\sigma'}(\cdot|y^{hl-j},y^j))\label{eq:last}\\
              & = & \sum_{y^j}   \sum_{y^{hl-j}}  p_{\sigma'}(y^{hl-j},y^j)  U(p_{\sigma'}(\cdot|y^{hl-j},y^j))\nonumber\\
              & = &\sum_{y^{hl}}  p_{\sigma'}(y^{hl})  U(p_{\sigma'}(\cdot|y^{hl})) \ =\ U(X|Y_{\sigma'})\nonumber
\end{eqnarray}

\noindent where the inequality \eqref{eq:last} follows from \eqref{eq:convex2}.
\end{proof}


%
%
%
In deterministic systems, repetitions of the same action  are not relevant: this   leads to the following improved upper bound on the length of the non-adaptive $\sigma'$ that simulates $\sigma$.

\begin{prop}\label{prop:buonddet} If the mechanism $\St$ is deterministic, then the upper-bound
in the previous theorem can be simplified to  $|\range(\sigma)|$.
\end{prop}
\begin{proof}
Let $\sigma$ be any finite non-adaptive strategy for $\St$. Suppose there is an action $a$ that occurs at least twice in $\sigma$,
seen as a tuple of actions, and let $\sigma_-$ be the non-adaptive strategy obtained by removing
the first occurrence of $a$ from $\sigma$, seen as  a list. Assume   the two   $a$'s occur  at position $i$ and $j$, $i<j$,
of $\sigma$.
 Since $\St$ is deterministic, it is easily seen that, for each $y^n=(y_1,\ldots,y_n)$, if $y_i\neq y_j $ then for each $x$
$p_\sigma(y^n|x)=0$ (as submitting twice the same action $a$ cannot give rise to two different answers $y_i$ and $y_j$), and as a consequence $p_\sigma(y^n)=0$. On the other hand, if $y_i=y_j$, then, denoting  by $y^{n-1}$   the sequence obtained by removing $y_i$
from $y^n$,  for each $x$ we have: $p_\sigma(y^n|x)= p_{\sigma_-}(y^{n-1}|x)$ (as $p(y_i|x)=p(y_j|x)$ is either 0 or 1), and as a consequence $p_\sigma(y^n)=p_{\sigma'}(y^{n-1})$ and
$p_\sigma(x|y^n)=p_{\sigma_-}(x|y^{n-1})$. This implies that $U(X|Y_\sigma)=U(X|Y_{\sigma_-})$. Repeating this elimination step, we can
 eventually  get rid of
 all the duplicates
in   $\sigma$, while preserving the value of $I_\sigma(\St,p)$. Applying this fact to the strategy
$\sigma'$ defined in the proof of Theorem \ref{th:cinque}, we can come up with a strategy $\sigma''$ of length
$|\range(\sigma)|$ such that $I_{\sigma''}(\St,p)= I_{\sigma_-}(\St,p)$.
\end{proof}

\begin{exa}\label{Ex-nonadapt-det}
{
We reconsider   Example  \ref{exDBdet}.
For the adaptive    strategy $\sigma$ defined in Figure  \ref{database3}, we have already shown that, for  $U(\cdot)=H(\cdot)$,   $I_{\sigma}(\St,p)\approx 2.45$.
Consider now the non-adaptive strategy   $\sigma'=[\mathrm{ZIP,Date,Age}]$, \fra{Io sarei per togliere questa constatazione "which is just one action longer.."  sulla differenza di lunghezza fra $\sigma$ e $\sigma'$. Solo in questo caso è così piccola e inoltre non riguarda il teorema applicato.
Forse poi dovremmo dire che qui applichiamo la Proposizione 1?}
which is just one action longer than $\sigma$.  The corresponding attack tree is reported in Figure  \ref{fig:DBprovatutto}: the final partition   obtained with $\sigma'$ is finer than the one obtained with $\sigma$.
In fact,  $I_{\sigma'}(\St,p)= \log{10}-\frac{2}{5}\approx 2.92 > I_{\sigma}(\St,p)$.}
\end{exa}

The  results discussed above are  important from the point of view of the analysis of randomization mechanisms. They entail that, for systems with a moderate number of actions, analyzing adaptive strategies is essentially equivalent to analyzing non-adaptive ones. The latter task can be much easier to accomplish. For example, results on asymptotic rate of convergence of non-adaptive strategies are available (e.g. \cite[Th. IV.3]{BP12}.) They permit  to analytically assess  the resistance of a mechanism  as the length of the considered strategies grows.

\ifmai
The following result, which covers the case of error entropy, is adapted from \cite[Th. IV.3]{BP12}, to which we refer for further details.   Assume $Act=\{a_1,\ldots,a_k\}$, and for each  $c_i\in \X/\equiv$, let $\pi_i\defin \max_{x\in c_i} p(x)$.

\begin{prop}\label{prop:iid}
For each $n\geq 1$, consider the non-adaptive strategy $\sigma^{(n)}=\defin [a_1,\ldots,a_k,\ldots,a_1,\ldots,a_k]$ ($n$ times) and let $(X,Y^{(n)})\sim p_{\sigma^{(n)}}(\cdot)$. Then, there are positive constants $\gamma$ and $\rho$, only depending on the   matrices
$p_a(\cdot|\cdot)$, ($a\in Act$), such that $(1-\sum_i \pi_i)+\gamma 2^{-n\rho}\geq E(X|Y^n)\geq  1-\sum_i \pi_i$.
\end{prop}
\fi
\section{Maximum Leakage}\label{sec:rif}
In this section we show that the class of adaptive and non adaptive strategies induce the same maximum leakage, where the maximum is taken over all strategies. For truly probabilistic mechanisms, strategies achieving maximum leakage are in general infinite. A key notion is that of indistinguishability:  an equivalence relation over $\X$ s.t. \fra{Avrei aggiunto questa scritta in rosso e tolto "equivalence" dalla Def per guadagnare un rigo}
$x$ and $x'$ are indistinguishable if,  no matter what strategy the adversary will play, he cannot tell them apart.

\begin{defi}[Indistinguishability]\label{def:indist} We define the following  equivalence relation
over $\X$:
$$x \equiv x' \quad \text{ iff } \  \text{ for each finite } \sigma:\, p_\sigma(\cdot|x) =  p_\sigma(\cdot|x')\,.$$
\end{defi}

Despite being based on a universal quantification over all finite strategies, indistinguishability is in fact quite easy to characterize, also computationally. For each $a\in Act$, consider the equivalence relation  defined by $x \equiv_{a} x'$ iff $p_a(\cdot|x)=p_a(\cdot|x')$. We have the following result (see the Appendix for a proof.)

\begin{lem}\label{lemma:due}
$x \equiv x'$ iff for each $a\in Act$, $p_a(\cdot|x) = p_a(\cdot|x')$. In other words, $\equiv$ is $\cap_{a\in Act}\equiv_a$.
\end{lem}

Now, consider $\X/\equiv$, the set of equivalence classes of $\equiv$, and let  $c$  ranges over this set.
Let $[X]$ be the r.v. whose outcome is the equivalence class of $X$ according
to $\equiv$. Note that $p(c)\defin \Pr([X]=c)=\sum_{x\in c} p(x)$. We consistently extend   our
$I$-notation by defining
$$ U(X\,|\,[X]) \defin  \sum_c p(c) U(p(\cdot|\,[X]=c)) \ \ \mbox{ and }\ \  I(X\, ;\,[X]) \defin  U(X)-U(X\,|\,[X])\,.$$
\noindent More explicitly, $p(\cdot|[X]=c)$ denotes the
distribution over $\X$ that yields $p(x)/p(c)$ for $x\in c$ and 0 elsewhere; we will often abbreviate $p(\cdot|[X]=c)$ just
as $p(\cdot|c)$. Note that $I(X\, ;\,[X])$ expresses the information gain about $X$ when the
attacker gets to know the indistinguishability class of the secret. As expected, this is an upper-bound to the information that can be gained by playing any strategy.


\begin{thm}\label{th:uno}
$I_\star(\St,p) \leq I(X\,;\,[X])$.
\end{thm}
\begin{proof}
 Fix any finite strategy $\sigma$ and prior $p(\cdot)$. It is sufficient to prove that $U(X|Y)\geq U(X\,|\,[X])$. The proof  exploits the concavity of $U$. First, we note that, for each $x$ and $y^j$ of nonzero probability we have ($c$ below ranges over $\X/\equiv$):
\begin{eqnarray}
p_\sigma(x|y^j) & = & \sum_c \frac{p_\sigma(x,y^j,c)}{p_\sigma(y^j)}\ = \ \sum_c p_\sigma(c|y^j) p_\sigma(x|y^j,c)\,.\label{eq:decomp}
 \end{eqnarray}
\noindent
By \eqref{eq:decomp}, concavity of $U(\cdot)$  and Jensen's inequality
\begin{eqnarray}
U(p(\cdot|y^j)) & \geq & \sum_c p_\sigma(c|y^j) U(p_\sigma(\cdot|y^j,c))\,.
\label{eq:J}
\end{eqnarray}
\noindent
Now, we can compute as follows (as usual, $y^j$ below runs over sequences of nonzero probability):
\begin{eqnarray}
U(X|Y) & = & \sum_{y^j} p_\sigma(y^j)U(p_\sigma(\cdot|y^j)) \ \ \geq \ \ \sum_{y^j,c}p_\sigma(y^j)p_\sigma(c|y^j)U(p_\sigma(\cdot|y^j,c))\label{eq:J2}\\
& = & \sum_{y^j,c}p_\sigma(y^j)p_\sigma(c|y^j)U(p(\cdot|c)) \ \ = \ \ \sum_c \left(\sum_{y^j}  p_\sigma(y^j,c)\right) U(p(\cdot|c))\label{eq:3}\\  
& = & \sum_c p(c) U(p(\cdot|c)) \ \ = \ \ U(X\,|\,[X])\nonumber
\end{eqnarray}
\noindent
where: \eqref{eq:J2} is justified by \eqref{eq:J}; and
the first equality in \eqref{eq:3} follows from the fact that, for each $x$, $p_\sigma(x|y^j,c)= p(x|c)$ (once the equivalence class of the secret is known, the observation $y^j$ provides no further information about the secret.)
\end{proof}

As to the maximal achievable information, we start our discussion from deterministic mechanism.
\begin{prop}\label{prop:determ} Let $\St$ be deterministic. Let $\sigma=[a_1,\ldots,a_k]$ be a   non-adaptive strategy that  plays all actions in $Act$ once. Then $I_\star(\St,p)=I_\sigma(\St,p)$.
\end{prop}
\proof
Let $(X,Y)\sim p_\sigma(\cdot)$. We prove that $U(X\,|\,Y) = U(X\,|\,[X])$. We first note that for each $c\in \X/\equiv$ there is exactly one sequence $y^k_c$   s.t. $p_\sigma(y^k_c|c)=1$: this follows from $\St$ being deterministic. Moreover, if $c\neq c'$ then $y^k_c\neq y^{k}_{c'}$: otherwise, it would follow  that $p_{a_i}(y|c)=p_{a_i}(y|c')$ for each $a_i\in Act$ and $y\in \Y$, contrary to Lemma \ref{lemma:due} (note that $p(\cdot|c)$ is the same as $p(\cdot|x)$, for any $x\in c$.) These facts can be used to show, through easy manipulations,    that $p(x|y^k_c)=p(x|c)$ for each $x$.
As a consequence, one can   compute as follows.
\begin{eqnarray*}
U(X|Y) & = & \sum_{y^k} p_\sigma(y^k)U(p_\sigma(\cdot|y^k))\\
       & = & \sum_c p(c) \sum_{y^k} p_\sigma(y^k|c) U(p_\sigma(\cdot|y^k))\\
       & = & \sum_c p(c)  U(p_\sigma(\cdot|y^k_c))\\
       & = & \sum_c p(c)  U(p_\sigma(\cdot|c))\\
       & = & U(X\,|\,[X])\,.\rlap{\hbox to 201 pt{\hfill\qEd}}
\end{eqnarray*}

\noindent Hence, in the deterministic case, the maximal gain in information is obtained by a trivial brute-force strategy where all actions are played  in any fixed order. It is instructive to observe such a strategy at work, under the form of an attack tree. The supports of the distributions that are at the same level constitute a partition of $\X$: more precisely,  the partition    at level $i$ ($1\leq i\leq k$) is given by the equivalence classes of  the relation $\cap_{j=1}^i \equiv_{a_j}$.
An example of this fact   is illustrated by the attack tree in Figure  \ref{fig:DBprovatutto}, relative to the non-adaptive strategy $[\mathrm{ZIP, Date, Age}]$ for the mechanism in Example  \ref{exDBdet}.
This fact had been already observed in \cite{KB07} for the restricted model considered there. Indeed,  one would obtain the model of \cite{KB07} by stripping the probabilities off the tree in Figure  \ref{fig:DBprovatutto}.

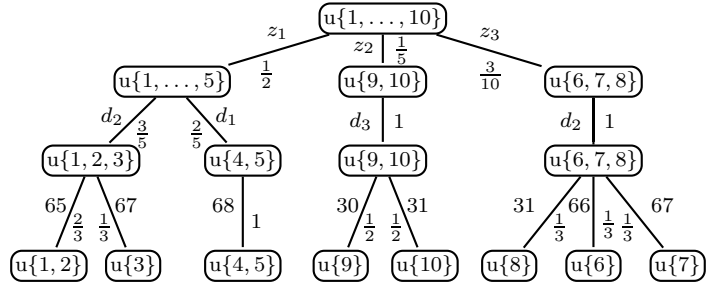
\begin{figure}[t]
\centering
\begin{tikzpicture}[-,>=stealth',shorten >=1pt,auto,node distance=2cm, thick, edge from parent path={(\tikzparentnode) -- (\tikzchildnode)}, scale=0.7]
  \tikzstyle{every state}=[rectangle,rounded corners, inner sep=1.5pt,minimum size=.4cm]
  \tikzstyle{level 1}=[level distance=1.2cm, sibling distance=4cm]
\tikzstyle{level 2}=[level distance=1.5cm, sibling distance=1.2cm]
\tikzstyle{level 3}=[level distance=2cm, sibling distance=1.6cm]
\node[state]{{\scriptsize $\unif \{1,\ldots,10\}$}}
 child
  {
  node [state]{{\scriptsize $\unif \{1,\ldots,5\}$ }}
     child {
     node [state, left=0.1pt] {{\scriptsize $\unif \{1,2,3\}$ }}
     child{ node [state] {{\scriptsize $\unif \{1,2\}$ }}
     edge from parent
          node [above=3pt, left=-3pt] {\scriptsize $65$}
          node[right=3pt, below=-4pt]  {\scriptsize  $\frac{2}{3}$} }
     child{ node [state] {{\scriptsize $\unif \{3\}$ }}
     edge from parent
          node [above=3pt, right=-3pt] {\scriptsize $67$}
          node[left=3pt, below=-4pt]  {\scriptsize  $\frac{1}{3}$} }
             edge from parent
                node[above=2pt, left=-1pt] {\scriptsize $d_2$}
                 node[below=5pt, right=-3pt]  {\scriptsize  $\frac{3}{5}$} }
      child { node [state, right=0.1pt] {{\scriptsize $\unif \{4,5\}$ }}
      child { node [state] {{\scriptsize $\unif \{4,5\}$ }}
             edge from parent
                node[above=3pt, left=-1pt] {\scriptsize $68$}
                 node[below=3pt,  right=-1pt]  {\scriptsize  $1$} }
             edge from parent
                node[above=2pt, right=-1pt] {\scriptsize $d_1$}
                 node[below=5pt, left=-3pt]  {\scriptsize  $\frac{2}{5}$} }
     edge from parent
                node[above] {\scriptsize $z_1$}
                node[below right]  {\scriptsize $\!\!\!\!\!\!\!\!\frac{1}{2}$}
   }
 child
  { node [state]    {{\scriptsize $\unif \{9,10\}$}}
     child {
     node [state] {{\scriptsize $\unif \{9,10\}$ }}
     child{ node [state] {{\scriptsize $\unif \{9\}$ }}
     edge from parent
          node [above=3pt, left=-3pt] {\scriptsize $30$}
          node[right=3pt, below=-4pt]  {\scriptsize  $\frac{1}{2}$} }
     child{ node [state] {{\scriptsize $\unif \{10\}$ }}
     edge from parent
          node [above=3pt, right=-3pt] {\scriptsize $31$}
          node[left=3pt, below=-4pt]  {\scriptsize  $\frac{1}{2}$} }
             edge from parent
                node[left] {\scriptsize $d_3$}
                 node[right]  {\scriptsize  $1$} }
    edge from parent
                node[above= 1pt, left=-1pt] {\scriptsize $z_2$}
                node[right]  {\scriptsize  $\frac{1}{5}$}
     {} }
  child
  { node [state]    {{\scriptsize $\unif \{6,7,8\}$}}
    child {
    node[state] {{\scriptsize $\unif \{6,7,8\}$ }}
    child{ node [state] {{\scriptsize $\unif \{8\}$ }}
     edge from parent
          node [above=3pt, left=-1pt] {\scriptsize $31\ $}
          node[right=3pt, below=-4pt]  {\scriptsize  $\frac{1}{3}$} }
     child{ node [state] {{\scriptsize $\unif \{6\}$ }}
     edge from parent
          node [above=3pt, left=-3pt] {\scriptsize $66$}
          node[below=-5pt]  {\scriptsize  $\quad\ \frac{1}{3}$} }
     child{ node [state] {{\scriptsize $\unif \{7\}$ }}
     edge from parent
          node [above=3pt, right=-1pt] {\scriptsize $\ 67$}
          node[left=3pt, below=-4pt]  {\scriptsize  $\frac{1}{3}$} }
             edge from parent
                edge from parent
                node[left] {\scriptsize $d_2$}
                 node[right]  {\scriptsize  $1$} }
  edge from parent
                node[above] {\scriptsize $z_3$}
                node[below]  {\scriptsize $\frac{3}{10}$}
    {}  };
\end{tikzpicture}
\caption{The attack tree corresponding to the the non-adaptive strategy $[\mathrm{ZIP, Date, Age}]$ for Example  \ref{exDBdet}.}
\label{fig:DBprovatutto}
\end{figure}

The general  probabilistic case
  is    more complicated. Essentially, any non-adaptive strategy where each action is played infinitely often achieves the maximum information gain. The next theorem considers one such strategy; the proof of this result is reported in Appendix \ref{app:proof}. 

\begin{thm}\label{th:due} There is a total, non-adaptive strategy $\sigma$ s.t. $I_\sigma(\St,p)=
 I(X;[X])$. Consequently, $I_\star(\St,p)=  I(X\,;\,[X])$.
\end{thm}

Of course, as shown in the preceding section,  finite adaptive strategies can be more efficient in terms of length by a factor of $|Act|$ when compared with non-adaptive ones.
Concerning capacity, we do not have a general formula for the maximizing distribution. In what follows, we limit our discussion to  two important cases for $U(\cdot)$, Shannon entropy and error entropy. In both cases, capacity only depends on the number $K$ of indistinguishability classes.
\mik{Controllare congettura G.E.}
For guessing entropy, we conjecture that $C(\St)=\frac{K-1}2$, but at the moment a proof of this fact escapes us. 

\begin{thm}\label{th:tre} The following formulae holds, where $K=|\X/\equiv|$.
\begin{itemize}
\item For $U=H$ (Shannon entropy), $C(\St)=\log K$.
\item For $U=E$ (Error entropy), $C(\St)=1-\frac{1}{K}$.
\end{itemize}
\end{thm}

\proof Let $x_i$ any representative of class $c_i$, for $i=1,\ldots,K$.
\begin{itemize}
\item $U=H$.
    By the symmetry  of mutual information in the case of Shannon entropy, we have
    \begin{eqnarray*}
    I(X; [X])& = & H([X])-\underbrace{H([X]\,|\,X)}_{=0}  \ \ \ =\ \ \ H([X])\\
             & = & -\sum_{c_i} p(c_i)\log p(c_i) \ \ \ \leq\ \ \ \log K
    \end{eqnarray*}

\noindent where the last inequality follows from the property of Shannon entropy that $H(q)\leq \log |\supp(q)|$, for any distribution $q$.
On the other hand, if we take the distribution $p(\cdot)$  defined as $p(x_i)=\frac 1 K$, and $p(x)=0$ elsewhere, we can easily compute that $I(X;[X])=\log K$.

\item $U=E$. Let $p(\cdot)$ be any prior and
assume without loss of generality that
$p(x_i)=\max_{x\in c_i} p(x)$ for each $i$, and furthermore that $p(x_1)=\max_x p(x)$.
By easy manipulations, we have:
    \begin{eqnarray*}
    I(X; [X])& = & E(X)- E(X\,|\,[X]) \\
             & = & (1-p(x_1)) -\left(1-\sum_{c_i} p(c_i)\frac{p(x_i)}{p(c_i)}\right)\\
             & = & \left(\sum_{i=1}^K p(x_i)\right)-p(x_1)\ \ \  = \ \ \ \sum_{i=2}^K p(x_i)\,.
    \end{eqnarray*}

\noindent Now it is easily checked that the last term in this chain is $\leq 1-\frac 1 K$: this is done by separately considering the two cases $\max_x p(x)=p(x_1)\leq \frac 1 K$ and $\max_x p(x)=p(x_1)> \frac 1 K$.  On the other hand, if we take, as above, the distribution $p(\cdot)$  defined as $p(x_i)=\frac 1 K$, and $p(x)=0$ elsewhere, we can easily compute that $I(X;[X])= 1-\frac 1 K$.
\ifmai
\item $U=G$.
Fix any prior $p(\cdot)$.
List the elements of $\supp(p)$, \emph{except} the one of least probability,
 in descending order of probability as $x_1,\ldots,x_N$ (hence $N =|\supp(p)|-1$; we ignore the trivial case where
 $|supp(p)=1$|.) Let $c_1,\ldots,c_H$ ($H\leq K$) be the classes whose intersection with the above set of $N$
 elements is nonempty.
For each such class  $c_i$,
 let ${i_j}$ denote the index of the
  the $j$-th element  of $c_i$ in the above ordering; that is,   $c_i $
  contains  $x_{i_1},\ldots,x_{i_{m_i}}\}$, for some
   $m_i\leq |c_i|$. We have the following equalities, which follow
    from easy manipulations of the indices in the summations.
    \begin{eqnarray}
    I(X; [X])& = & G(X)- G([X]|X) \\
             & = & \sum_{j=1}^{N} j p(x_j) -\sum_{i=1}^H p(c_i)\sum_{l=1}^{m_i-1} l\frac{p(x_{i_l})}{p(c_i)}\\
             & = & \sum_{i=1}^H \sum_{l=1}^{m_i-1} p(x_{i_l})(i_l-l)\,+\,\sum_{i=1}^H p(x_{i_{m_i}})
             i_{m_i}\,.\label{eq:guessing}
    \end{eqnarray}
%
\ifmai
Now, consider the class containing $x_N$; for notational simplicity, assume this class is $c_1$: $x_N=x_{1_{m_1}}$.
Consider the new distribution obtained from $p(\cdot)$ by    moving  the probability
 mass of   $X_N$ into $x_{1_1}$, thus:
  $p'(x_{1_1})\defin p(x_{1_1})+ p(x_{N})$, $p'(x_{N})\defin 0$ ($p'(\cdot)$ coincides with $p(\cdot)$
   elsewhere.)
Now, assume that there is at least one class $c_i$ that such that $m_i>1$ (if not, skip to
the next paragraph.) Among
those classes,
take the one containing the element of minimal index. Assuming, for notational simplicity, that
this is class $c_1$, we have therefore: in terms of indices, $1_1> i_1$ and
$p(x_{1_1})\geq p(x_{i_1})$, for all  $i\neq 1$ such that $m_i>1$.
Consider now a new distribution obtained from $p(\cdot)$ by    moving  the probability
 mass of the element with greatest index in $c_1$ to $x_{1_1}$, thus:
  $p'(x_{1_1})\defin p(x_{1_1})+ p(x_{1_{m_i}})$, $p'(x_{1_{m_i}})\defin 0$ ($p'(\cdot)$ coincides with $p(\cdot)$
   elsewhere.)
  Consider the elements of in the support of $p'(\cdot)$, ordered as before in terms of decreasing
  probability: $x'_1,\ldots,x'_{N-1}$; denote by $x_{i'_j}$ the the $j$-th element of class $c_i$ in this order.
    We note the following.
  \begin{enumerate}
  \item  each class $c_i$ different from $c_1$ has still the same set of elements, $x'_{i'_1},\ldots,x'_{i'_{m_i}}$,
  but in terms of indices, $i'_j\leq i_j$, since due to the elimination of an element in class $c_1$,
   some elements in the new order  may have shifted a   position up.
  \item and the resulting $I(X';[X'])$, as given by formula \eqref{eq:guessing}.
There are now two cases. If $\max_x p(x)\leq \frac 1 K$,
from \eqref{eq:guessing} it follows
\[
I(X;[X]) \leq \frac{N\cdot (N+1)} {2K} - \sum_{i=1}^K \frac{l_i(l_i+1)}{2K}\,.
\]
\fi
\ifmai
\noindent Now, $I(X;[X])$ as a function of $p(\cdot)$ is continuous, hence it achieves a maximum in the compact set $\P$, on some element $p^\star(\cdot)$. Let $V$ be the convex subset   generated by $\supp(p^\star)$. Function  $I(X;[X])$ restricted to $V$ is still continuous and achieves a maximum at $p^\star(\cdot)$, which is an  internal point of $V$ (we  exclude the trivial case where $|\supp(p^\star)|=1$.) Since $I(X;[X])$ is   differentiable, this implies that its partial derivatives at each $x_i\in \supp(p^\star)$, evaluated at $x_i=p(x_i)$, must vanish. We take the partial derivatives of the term in \eqref{eq:guessing} for each $x_i$ ,   equate them to 0, and obtain:
\[
i_l = l \quad\text{ for each } i=1,\ldots,K, \quad l=1,\ldots,l_i\,.
\]
Now, this can only be the case if for each class $c_i$ there is precisely one element $x_i\in \supp(p)$.
\fi
\fi\qed
\end{itemize}

\begin{exa}{
Consider the mechanism defined in Example  \ref{exDBdet}. One has the following capacities:
for $U(\cdot)=H(\cdot)$, $C(\St)=\log{8}=3$, while for $U(\cdot)=E(\cdot)$, $C(\St)=\frac{7}{8}=0.875$.}
\end{exa}


\section{Computing Optimal Finite Strategies}\label{sec:finitestrat}
We show that, for finite strategies,   $I_\sigma(\St,p)$
can be expressed recursively as
a Bellman  equation. This allows for calculation of  {optimal}
finite strategies  based on standard algorithms, such as   backward induction.

\subsection{A Bellman Equation}
Let us introduce some terminology.
For each  $y$, the $y$-\emph{derivative} of $\sigma$, denoted $\sigma_{y}$, is the function  defined thus, for each $y^j\in\Y^*$: $\sigma_{y}(y^j)\defin \sigma(yy^j)$. Note that if $\sigma$ has length $l>1$, then $\sigma_y$ is a strategy  of height $\leq l-1$. For $l=1$, $\sigma_y$ is the empty function.
Recall that according to \eqref{eq:update}, for $h=ay$, we have\footnote{In terms of a given prior $p(\cdot)$ and of the matrices of $\St$, this can be also expressed as: $p^{ay}(x)=\frac{p_a(y|x)p(x)}{\sum_{x'}p_a(y|x')p(x')}$.}
\begin{eqnarray}
\nonumber
p^{ay}(x) & = & p_a(x|y)\,.
\end{eqnarray}

\noindent By convention, we let $I_\sigma(\cdots)$ denote 0 when $\sigma$ is empty. Moreover, we write  $I_{[a]}(\cdots)$ as  $I_{a}(\cdots)$.

\begin{lem}\label{lemma:decompo}
Let $p(\cdot)$ be any prior on $\X$. Let $\sigma$ be  a  strategy  with $\sigma(\varepsilon)=a$. Then
 $I_\sigma(\St;p) = I_{a}(\St;p)+\sum_y p_a(y) I_{\sigma_y}(\St;p^{ay})$.
\end{lem}

We introduce some additional notation to be used in the proof of this lemma. Let $l$ denote the   length of a strategy $\sigma$,  and let  $(X,Y)$ be distributed according to $p_\sigma(\cdot)$. We can decompose $Y$ as the concatenation of the 1st observation and whatever sequence of observations is left, thus: $Y= Y_1\cdot Y_s$. Here, $Y_1$ takes values on $\Y$, while $Y_s$ takes values on a   subset of $\cup_{1\leq j\leq l}\Y^j$ - in particular, if $l=1$, $Y_s$ takes on the value $\varepsilon$ with probability 1.  In what follows, we denote  the marginal distribution of $Y_1$ under $\sigma$ just as   $p_\sigma(y)$, and that of $Y_s$ as $p_\sigma(y^j)$, for generic $y$ and $y^j$.

\proof{(of Lemma \ref{lemma:decompo})} 
It is an easy matter to prove the following equations, for each prior $p(\cdot)$, finite strategy $\sigma$ with $\sigma(\varepsilon)=a$, sequence $y^j$, observation $y$, one has (below, $y$ and $y^j$ run over elements of nonzero probability; moreover, for any prior $p(\cdot)$, history $h$ and strategy $\sigma$, the term $p^h_\sigma$ is to be parsed as $(p^h)_\sigma$):
\begin{eqnarray}
p_\sigma(y) & = & p_a(y)\label{eq:uno}\\
p_\sigma(x|y) & = & p_a(x|y) \; = \; p^{ay}(x)\label{eq:due}\\
p_\sigma(x|yy^j) & = &    p^{ay}_{\sigma_y}(x|y^j)\label{eq:tre}\\
p_\sigma(y^j|y) & = & p^{ay}_{\sigma_y}(y^j)\,.  \label{eq:upddue}
\end{eqnarray}

\noindent By applying   equalities \eqref{eq:uno}, \eqref{eq:due}, \eqref{eq:tre} and  \eqref{eq:upddue} above as appropriate, we have:
\begin{align*}
I_\sigma(\St,p)\ &= \  I(X;Y)\\
                & =  \left[U(X)-U(X|Y_1)\right] \;+\; \left[U(X|Y_1) - U(X|Y)\right]\\
                & =  \left[U(p)-\sum_{y} p_\sigma(y)U(p_\sigma(\cdot|y))\right]\;+\;
                \left[ \sum_{y} p_\sigma(y)   U(p_\sigma(\cdot|y)) - p_\sigma(y,y^j)U(p_\sigma(\cdot|yy^j))\right]\\
                & =  \left[U(p)-\sum_{y} p_\sigma(y)U(p_\sigma(\cdot|y))\right]\;+\;
                \left[ \sum_{y} p_\sigma(y)   U(p_\sigma(\cdot|y)) - p_\sigma(y)p_\sigma(y^j|y)U(p_\sigma(\cdot|yy^j))\right]\\
                & =  \left[U(p)-\sum_{y} p_\sigma(y)U(p_\sigma(\cdot|y))\right]\;+\;
                \sum_{y} p_\sigma(y) \left[ U(p_\sigma(\cdot|y)) - \sum_{y^j}p_\sigma(y^j|y)U(p_\sigma(\cdot|yy^j))\right]\\
                & =  \left[U(p)-\sum_{y} p_{a}(y)U(p_{a}(\cdot|y))\right]\;+\;
                     \sum_{y} p_a(y) \left[ U(p^{ay}) - \sum_{y^j} p^{ay}_{\sigma_{y}}(y^j)U(p^{ay}_{\sigma_{y}}(\cdot|y^j))\right]\\
                & =  I_{a}(\St;p)\;+\;
                      \sum_{y} p_a(y)
                  I_{\sigma_y}(\St;p^{ay}.)\rlap{\hbox to 228 pt{\hfill\qEd}}
\end{align*}

\noindent Let us say that a   strategy $\sigma$ of length $l$ is \emph{optimal} for $\St$, $p(\cdot)$ and $l$ if it maximizes $I_\sigma(\St,p)$ among all strategies of length $l$.

\begin{cor}[Bellman-type equation for optimal strategies]\label{cor:Bellman}
There is an optimal  strategy $\sigma^\star$ of   length $l$ for $\St$ and $p(\cdot)$ that  satisfies the following equation
\begin{equation}\label{eq:bellman}
I_{\sigma^\star}(\St;p) = \max_a \; \left\{I_{a}(\St;p)\;+\;
                      \sum_{y:\,p_a(y)>0} p_a(y) I_{\sigma^\star_{a,y}}(\St;p^{ay}) \,\right\}\,
\end{equation}

\noindent where $\sigma^\star_{a,y}$ is an optimal   strategy of length $l-1$ for $\St$ and $p^{ay}(\cdot)$.
\end{cor}

Corollary \ref{cor:Bellman}  allows us to employ dynamic programming  or backward induction to compute  {optimal} finite strategies. We discuss this briefly in the next subsection.


\subsection{Markov Decision Processes and Backward Induction}\label{sec:MDP}
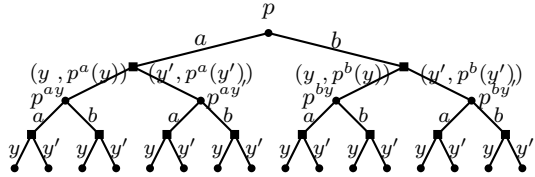
\begin{figure}
\vspace{-15pt}
\tikzstyle{bag} = [text centered,circle, minimum width=3pt,fill, inner sep=0pt]
\tikzstyle{Node}=[draw, minimum width=1.5pt,fill, inner sep=1.5pt]
\centering
\begin{tikzpicture}[scale=0.9]
\draw (0,0) node [bag, label=above:{\footnotesize$p$}] {};
\draw [thick] (0,0) to (-2,-0.5);
\draw [thick] (0,0) to (2,-0.5);
\draw (-2,-0.5) node [Node] {};
\draw (2,-0.5) node [Node] {};
\draw (-1,-0.1) node {\footnotesize $a$};
\draw (1,-0.1) node {\footnotesize $b$};
\draw [thick] (-2,-0.5) to (-3,-1);
\draw [thick] (-2,-0.5) to (-1,-1);
\draw (-3,-1) node [bag] {};
\draw (-1,-1) node [bag] {};
\draw (-2.8,-0.6) node {\scriptsize $(y\phantom{'},p^a(y))$};
\draw (-1,-0.6) node {\scriptsize $(y',p^a(y'))$};
\draw (-3.2,-0.9) node {\footnotesize $p^{ay\phantom{'}}$};
\draw (-0.6,-0.9) node {\footnotesize $p^{ay'}$};

\draw [thick] (2,-0.5) to (3,-1);
\draw [thick] (2,-0.5) to (1,-1);
\draw (3,-1) node [bag] {};
\draw (1,-1) node [bag] {};
\draw (3,-0.6) node {\scriptsize $(y',p^b(y'))$};
\draw (1.1,-0.6) node {\scriptsize $(y\phantom{'},p^b(y))$};
\draw (0.8,-0.9) node {\footnotesize $p^{by\phantom{'}}$};
\draw (3.4,-0.9) node {\footnotesize $p^{by'}$};

\draw [thick] (-3,-1) to (-3.5,-1.5);
\draw [thick] (-3,-1) to (-2.5,-1.5);
\draw [thick] (-1,-1) to (-1.5,-1.5);
\draw [thick] (-1,-1) to (-0.5,-1.5);
\draw (-3.5,-1.5) node [Node] {};
\draw (-2.5,-1.5) node [Node] {};
\draw (-0.5,-1.5) node [Node] {};
\draw (-1.5,-1.5) node [Node] {};
\draw (-3.4,-1.25) node {\scriptsize $a$};
\draw (-2.6,-1.2) node {\scriptsize $b$};
\draw (-1.4,-1.25) node {\scriptsize $a$};
\draw (-0.6,-1.2) node {\scriptsize $b$};

\draw [thick] (3,-1) to (3.5,-1.5);
\draw [thick] (3,-1) to (2.5,-1.5);
\draw [thick] (1,-1) to (1.5,-1.5);
\draw [thick] (1,-1) to (0.5,-1.5);
\draw (3.5,-1.5) node [Node] {};
\draw (2.5,-1.5) node [Node] {};
\draw (0.5,-1.5) node [Node] {};
\draw (1.5,-1.5) node [Node] {};
\draw (3.4,-1.2) node {\scriptsize $b$};
\draw (2.6,-1.25) node {\scriptsize $a$};
\draw (1.4,-1.2) node {\scriptsize $b$};
\draw (0.6,-1.25) node {\scriptsize $a$};

\draw [thick] (-3.5,-1.5) to (-3.75,-2);
\draw [thick] (-3.5,-1.5) to (-3.25,-2);
\draw [thick] (-2.5,-1.5) to (-2.75,-2);
\draw [thick] (-2.5,-1.5) to (-2.25,-2);
\draw [thick] (-1.5,-1.5) to (-1.75,-2);
\draw [thick] (-1.5,-1.5) to (-1.25,-2);
\draw [thick] (-0.5,-1.5) to (-0.75,-2);
\draw [thick] (-0.5,-1.5) to (-0.25,-2);
\draw (-3.75,-2) node [bag] {};
\draw (-3.25,-2) node [bag] {};
\draw (-2.75,-2) node [bag] {};
\draw (-2.25,-2) node [bag] {};
\draw (-1.75,-2) node [bag] {};
\draw (-1.25,-2) node [bag] {};
\draw (-0.75,-2) node [bag] {};
\draw (-0.25,-2) node [bag] {};
\draw (-3.7,-1.7) node {\scriptsize $y\phantom{'}$};
\draw (-3.2,-1.7) node {\scriptsize $y'$};
\draw (-2.7,-1.7) node {\scriptsize $y\phantom{'}$};
\draw (-2.2,-1.7) node {\scriptsize $y'$};
\draw (-1.7,-1.7) node {\scriptsize $y\phantom{'}$};
\draw (-1.2,-1.7) node {\scriptsize $y'$};
\draw (-0.7,-1.7) node {\scriptsize $y\phantom{'}$};
\draw (-0.2,-1.7) node {\scriptsize $y'$};

\draw [thick] (3.5,-1.5) to (3.75,-2);
\draw [thick] (3.5,-1.5) to (3.25,-2);
\draw [thick] (2.5,-1.5) to (2.75,-2);
\draw [thick] (2.5,-1.5) to (2.25,-2);
\draw [thick] (1.5,-1.5) to (1.75,-2);
\draw [thick] (1.5,-1.5) to (1.25,-2);
\draw [thick] (0.5,-1.5) to (0.75,-2);
\draw [thick] (0.5,-1.5) to (0.25,-2);
\draw (3.75,-2) node [bag] {};
\draw (3.25,-2) node [bag] {};
\draw (2.75,-2) node [bag] {};
\draw (2.25,-2) node [bag] {};
\draw (1.75,-2) node [bag] {};
\draw (1.25,-2) node [bag] {};
\draw (0.75,-2) node [bag] {};
\draw (0.25,-2) node [bag] {};
\draw (3.8,-1.7) node {\scriptsize $y'$};
\draw (3.3,-1.7) node {\scriptsize $y\phantom{'}$};
\draw (2.8,-1.7) node {\scriptsize $y'$};
\draw (2.3,-1.7) node {\scriptsize $y\phantom{'}$};
\draw (1.8,-1.7) node {\scriptsize $y'$};
\draw (1.3,-1.7) node {\scriptsize $y\phantom{'}$};
\draw (0.8,-1.7) node {\scriptsize $y'$};
\draw (0.3,-1.7) node {\scriptsize $y\phantom{'}$};
\end{tikzpicture}
\caption{{\small The first few levels of a \mdp\ induced by a prior $p(\cdot)$ and a mechanism with $Act=\{a,b\}$ and $\Y=\{y,y'\}$.
Round nodes are   decision nodes and squares nodes are probabilistic nodes. For the sake of space, labels of the last level of arcs and nodes are only partially  shown.
}}
\label{fig:MDP}
\end{figure}

A mechanism $\St$ and a prior $p(\cdot)$ induce a     \emph{Markov Decision Process} (\mdp), where all possible attack trees are represented at once. Backward induction amounts to recursively computing the   most efficient attack tree out of this \mdp, limited to a given length.
More precisely, the \mdp\  $\M$ induced by $\St$ and a prior $p(\cdot)$  is an in general infinite tree consisting of   \emph{decision} nodes and  \emph{probabilistic} nodes. Levels  of decision nodes alternate with levels of probabilistic nodes,  starting from the root which is a decision node. Decision nodes are labelled with probability distributions over $\X$, edges outgoing decision nodes with actions, and edges outgoing probabilistic nodes with pairs $(y,\lambda)$ of an observation and a real, in such a way that (again, we identify nodes with the corresponding history):
\begin{itemize}
\item a decision node corresponding to history $h$ is labelled with $p^h(\cdot)$, if this is defined, otherwise the node and its descendants are removed, as well as the incoming edge;
\item for any pair of consecutive edges leading from a decision node $h$   to another decision node $hay$, for any $a\in Act$ and $y\in \Y$, the edge outgoing the probabilistic node is labelled with $(y,p^h_a(y))$.

Figure  \ref{fig:MDP} shows the first   few levels of such a  \mdp.
\end{itemize}

%
%
%
\noindent In order to compute an optimal strategy of length $l\geq 1$ by backward induction, one   initially prunes the tree at $l$-th decision level (the root is at level 0) and then assigns  \emph{rewards}  to all leaves of the resulting  tree. Moreover, each probabilistic node is assigned an \emph{immediate gain}.  Rewards are then gradually propagated from the leaves up to the root,   as follows:
\begin{itemize}
\item each probabilistic node is assigned as a reward the sum of its immediate gain and the \emph{average} reward  of its children, average computed using the   probabilities on the outgoing arcs;
\item each decision node is assigned the \emph{maximal} reward of its children; the arc leading to the maximizing child  is marked or otherwise recorded.
\end{itemize}
Eventually, the root will be assigned   the maximal achievable reward. Moreover,   the paths of marked arcs starting from the root will define an optimal strategy  of length $l$. We can apply this strategy to our problem, starting   with  assigning rewards  $0$ to each leaf node $h$, and immediate gain $I_a(\St,p^h)$  to each $a$-child of any  decision node $h$. The correctness of the  resulting procedure is obvious in the light of Corollary \ref{cor:Bellman}.\bigskip

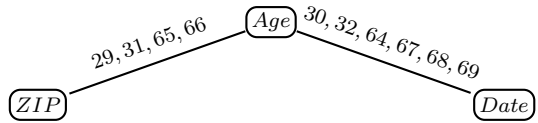
\begin{figure}[h]
\centering
\vspace{-12pt}
\begin{tikzpicture}[-,>=stealth',shorten >=1pt,auto,node distance=2cm, thick, edge from parent path={(\tikzparentnode) -- (\tikzchildnode)}]
  \tikzstyle{level 1}=[level distance=1.1cm, sibling distance=6.2cm]
  \tikzstyle{every state}=[rectangle,rounded corners, inner sep=2pt,minimum size=.4cm]
\node[state] {{\scriptsize $Age$}}
 child
  { node [state]    {{\scriptsize $ZIP$}}
  edge from parent
  node [above,sloped] {\scriptsize $29,31,65,66$}}
  child
{ node [state]    {{\scriptsize $Date$}}
  edge from parent
  node [above,sloped] {\scriptsize $30,32,64,67,68,69$}}
;
\end{tikzpicture}
\caption{A Shannon entropy optimal strategy for Example  \ref{exDB}.
Leaves with the same label and their incoming arcs have been coalesced.}\label{optimalstrategy}
\end{figure}
In a crude implementation of the above outlined procedure, the number of decision nodes in the \mdp\ will be bounded by $(|\Y| \times |Act|)^{l+1}-1$ (probabilistic nodes can be dispensed with, at the cost of moving incoming action labels to   outgoing arcs.) Assuming that each distribution  is stored in space $O(|\X|)$, the \mdp\ can be built and stored in time and space $\mathcal{O}(|\X| \times (|\Y| \times |Act|)^{l+1})$. This is also the running time of the backward induction outlined above, assuming $U(\cdot)$ can be computed in time
$O(|\X|)$ (some straightforward optimizations are possible here, but we will not dwell on this.) By comparison, the running time of the exhaustive procedure for deterministic systems outlined in \cite[Th.1]{KB07} is
$O(l \times |Act|^{r^l}\times |\X|\times\log |\X|)$, where $r$ is the maximal number of   classes in any relation $\equiv_a$; since $r$ can be as large as $|\Y|$, this gives  a worst-case running time of $O(l\times |Act|^{|\Y|^l} \times|\X|\times\log |\X|)$.

\begin{exa}\label{ex:optimal}
Applying backward induction to the mechanism of  Example  \ref{exDB} with $U(\cdot)=H(\cdot)$ and  $l=2$, one gets the optimal strategy  $\sigma$ shown in Figure  \ref{optimalstrategy}, 
\fra{ho sostituito "that yields" con "with" per risparmiare un rigo}
with  $I_\sigma(\St,p)\approx 2.4$ bits. Some details of the derivation of this strategy are reported in Appendix \ref{app:backward}.
\end{exa}


\section{The role of concavity}\label{sec:roleconcave}
We elucidate a connection between our definition of uncertainty   (convexity + continuity) and a  concept of scoring rule in Bayesian  decision  theory. 
A scoring rule encodes a system of (dis)incentives:  a   wrong forecast about an event causes the forecaster a  loss, whose magnitude   depends on  both the forecast that has been put forward, and on the event that has actually occurred. The  average loss under the best forecast is named entropy in this context. In essence,  we will show that: (a) every proper scoring rule induces an entropy function that is concave, hence necessarily continuous at least in the interior of the probability simplex; (b)   every concave function arises, under an additional mild assumption, as the entropy induced by  a certain scoring rule. This almost complete correspondence, and its simple definition, give a strong support to our choice of the class of uncertainty functions.

The connection between concavity and uncertainty has been explored in Statistics at least starting from the 1950's, and it comes into many different flavours. The following discussion is our personal take of this issue, for which we claim no technical novelty. Our presentation  is partly  inspired by \cite[Sections 9-10]{Dawid}.  For notational simplicity, in what follows we fix any ordering of the elements of $\X$, say $x_1,...,x_n$, so that we can identify  any distribution $p$ with a vector $(p_1,...,p_n)\in \mathbb{R}^n$.

Formally, in the context of  Bayesian decision theory, a \emph{scoring rule} is a function
\[
S:\X\times \mathring \P \rightarrow \mathbb{R}
\]
where  $\mathring\P$ denotes here the interior part of $\P$, that is, the set of  those distributions $q$ s.t. $q(x)>0$ for each $x\in \X$. This function  is given the following interpretation. A forecaster (in our case, the adversary) is asked to put forward a forecast  about the outcome of an event (in our case, the secret.) A  forecast takes the form of a probability distribution $q$, which represents the forecaster's estimation of the probability of each possible outcome  $x\in \X$. Then    $S(x,q)$ represents the  \emph{loss} incurred by the forecaster when the outcome is actually $x$ and he has put forward $q$. If the outcome is  distributed according to $p$, the average loss incurred if putting forward $q\in \mathring\P$ is given by
\[
S(p,q)\defin  \sum_x p(x) S(x,q)\,.
\]
This definition  is extended to each $q$ in the frontier,   $q\in \P\setminus\mathring\P$, by $\liminf_{q'\rightarrow q} S(p,q')$, be this finite or infinite.    The scoring rule $S$ is called \emph{proper} if the choice   $q=p$ always minimizes $S(p,q)$. In other words, a proper scoring rule (PSR) encodes  a penalty system  that forces the forecaster to be honest and propose the distribution he really thinks is the true one. For a  PSR, the   minimal loss corresponding to $p$ is - perhaps not surprisingly - also called the induced \emph{entropy}
\[
H(p)\defin S(p,p)\,.
\]
Seen as a function of  $p$, $H(p)$  gives a measure of the intrinsic risk associated with each $p$, under the   loss model encoded by the PSR $S$.

\begin{prop}\label{prop:conc1} Any entropy function  induced  by a PSR is a concave function over  $\mathcal{P}$. As a consequence, it is continuous over $\mathring \P$.
\end{prop}
\begin{proof}
Let $H(p)=S(p,p)$, with $S$ a PSR. Consider any $\lambda\in [0,1]$ and $p,q\in \P$, and let $r=\lambda p +(1-\lambda)q$. Assume first that $r$ is in the interior of $\P$. Then
\begin{eqnarray*}
H(r) & = & S(r,r)\\
    & =  & S(\lambda p +(1-\lambda)q, r) \\
    & = & \lambda S(p,r)+(1-\lambda) S(q,r)\\
    & \geq &  \lambda S(p,p)+(1-\lambda) S(q,q)\\
    & = &  \lambda H(p)+(1-\lambda) H(q)
\end{eqnarray*}
where the second equality follows from the linearity of $S$ w.r.t. the first argument and the inequality follows from the definition of PSR. If $r$ is in the frontier of $\P$,
by the properties of  $\liminf$, one has
\begin{eqnarray*}
H(r) & = & \liminf_{r'\rightarrow r} S(r,r')\\
    & = &  \liminf_{r'\rightarrow r}  S(\lambda p +(1-\lambda),r')\\
    & = &   \liminf_{r'\rightarrow r}  \lambda S(p,r')+(1-\lambda) S(q,r')\\
    & \geq &   \lambda \liminf_{r'\rightarrow r}  S(p,r')+ (1-\lambda)    \liminf_{r'\rightarrow r}   S(q,r')\\
    & = & \lambda S(p,r)+ (1-\lambda)  S(q,r)
\end{eqnarray*}
and then the reasoning proceeds as above. This concludes the proof that $H$ is concave.

Finally, it is a standard result that concavity over $\P$ implies continuity -- in fact, local lipschitzianity - over $\mathring\P$ (see e.g. \cite{ConvexOpt}.)
\end{proof}

On the other hand, under a mild additional assumption, any concave function   is induced by  a PSR, as we will check shortly. Let $H$ be concave on $\mathcal{P}$. It is well known that concave functions enjoy the following \emph{supporting hyperplane property} (see e.g. \cite{ConvexOpt}.) For each $q\in  \mathring\P$, there exists a vector
 $c_q=(c_1,...,c_n)\in \mathbb{R}^{n}$ such that, for each $p\in   \P$ (here $\langle\cdot,\cdot\rangle$ denotes the usual scalar product between two vectors)
\begin{eqnarray}\label{eq:supporting}
H(p) & \leq & H(q)+  \langle c_q,  p-q \rangle  \;=\;H(q)+\sum_{i=1}^n c_i (p_i-q_i)\,.
\end{eqnarray}
The above relation merely means that the graph of $H$ is all below an  hyperplane that is tangent to the point $(q,H(q))$. A vector $c_q$ satisfying \eqref{eq:supporting} is called a \emph{subgradient} of $H$ at $q$.  In particular, if $H$ is differentiable at $q$, there is exactly one choice for $c_q$, namely the gradient of $H$ at $q$
\begin{eqnarray*}
 c_q & = & \nabla H (q) \,=\,\left(\frac{\partial H}{\partial x_1}(q),...,\frac{\partial H}{\partial x_n}(q)\right)\,.
 \end{eqnarray*}
Now, for $q\in\mathring \P$, we set
\begin{eqnarray}\label{eq:psr}
S(x_i,q) & \defin &  H(q) + c_i -\langle c_q,  q \rangle\,.
\end{eqnarray}
Clearly, for each $p$, we have  $S(p,q)= E_p[S(X,q)]= H(q)+ \langle c_q, p-q\rangle $. We extend this to $q$ in the frontier by $S(p,q)\defin \liminf_{q'\rightarrow q} S(p,q')$. Assume now that the vectors $c_q$ can be chosen in such a way that the following property holds true for any $p$ in the frontier
\begin{eqnarray}\label{eq:lim}
\liminf_{p'\rightarrow p} \langle c_{p'}, p-p'\rangle & = & 0\,.
\end{eqnarray}
This condition means that, as $p'$ approaches a point $p$ in the frontier, the tangent hyperplane at $p'$ does not approach  a vertical hyperplane "too fast".  We now check that $S(x,q)$ is a PSR.
First, note that  for any $p\in\P$,
 we have $S(p,p)= H(p)$. For $p$ in the interior, this follows by definition; while for $p$ in the frontier, we note that $S(p,p)=\liminf_{p'} S(p,p')= H(p)+\liminf_{q'} \langle c_{p'}, p-p'\rangle$ and then exploit \eqref{eq:lim}. Now, applying   the supporting hyperplane property \eqref{eq:supporting} for $q$ in the interior we have
\begin{eqnarray*}
 S(p,q) & = &  E_p[S(X,q)] \\
        & = &  H(q)+ \langle c_q, p-q\rangle \\
        & \geq & H(p)\\
        & = & S(p,p)\,.
\end{eqnarray*}
For $q$ in the frontier, the same property  follows
from \[S(p,q)=\liminf_{q'}S(p,q')\geq \liminf_{q'}S(p,p)=S(p,p)\,,\] where the inequality follows from above.
This shows that $S(x,q)$ is a PSR and that the induced entropy is precisely the given concave function $H$. In other words, we have just shown  the following proposition.

\begin{prop}\label{prop:concave2} Every concave function over $\P$ that respects \eqref{eq:lim} is the entropy induced by some PSR.
\end{prop}

\begin{exa}
Let us consider the Shannon entropy function $H(p)=-\sum_i p_i\log p_i$. By direct calculation, it is immediate to check that $S(x_i,q)=-\log q_i$ is a PSR for Shannon entropy. However, it is instructive to apply the recipe in the proof of Proposition \ref{prop:concave2} to reconstruct $S(x,q)$.

$H(p)$ is differentiable in the interior of the probability simplex.  
The gradient of $H$ at $q$ is
\[
\nabla H(q)= \left(-(1+\log q_1),...,-(1+\log q_n)\right) \defin c_q\,.
\]
Now easy calculations show that $\langle c_{p'},p-p'\rangle=H(p)-H(p')+D(p||p')\rightarrow 0$ as $p'\rightarrow p$, thus making the condition \eqref{eq:lim} true\footnote{Here $D(p||p')\defin \sum_i p_i\log({p_i}/q_i)$ is the familiar Kullback-Leibler divergence. Note in particular that $D(p||p')\rightarrow 0$ as $p'\rightarrow p$. In passing, it is not true in general that $D(p||p')\rightarrow 0$ as $p\rightarrow p'$.}.    Noting that $\langle c_q,q\rangle = H(q)-1$, we apply \eqref{eq:psr}  and define
\[
S(x_i,q)\defin H(q)-(1+\log q_i)-(H(q)-1)\,=\, -\log q_i\,.
\]
A similar calculation for $H(p)=\var(p)$ yields
\[
S(x_i,q)=\var(q)+x^2_i-2x_i-E_q[X^2]+2E_q[X] = (x_i-1)^2-(E_q[X]-1)^2
\]
which leads to $S(p,q)=-E_q[X]^2+2E_q[X]+E_p[X^2]-2E_p[X]$ and, as expected, $S(p,p)=\var(p)$.

We finally consider the case of error entropy, $H(p)=1-\max_i p_i$. For each $q$ let us denote by $j_q$  an index in $\{1,...,n\}$ such that $q_{j_q}= \max_i q_i$ (if there is more than one such index, choose one arbitrarily.) At each $q$  with a  single maximal element, $H$ is differentiable and we can  set (here $\delta_{xy}$ is the Kronecker's delta symbol)
\[
c_q \defin \nabla H(q) = (-\delta_{1j_q},...,-\delta_{nj_q})\,.
\]
Hence, applying \eqref{eq:psr}  and noting that $\langle c_q,q\rangle= -q_{j_q}$, for each such $q$ we let
\[
S(x_i,q) \defin (1-q_{j_q})-\delta_{ij_q}+q_{j_q} = 1-\delta_{ij_q}\,.
\]
At points $q$ with more than one  maximal element,  $H$ is not differentiable, and the choice of $c_q$ is not unique. However,    $c_q=(-\delta_{1j_q},...,-\delta_{nj_q})$  is still a subgradient, so the same definition  of $S(x_i,q) $  as above applies. Note that $S(p,q)=1-p_{j_q}$ and, in particular, $S(p,p)=1-p_{j_p}=1-\max_i p_i$ as expected.
\end{exa}

\section{Conclusion, Related and Further Work}\label{sec:concl}
We have proposed a general information-theoretic model for the analysis of confidentiality under adaptive attackers. Within this model, we have proven several results on the limits of such attackers,   on the relations between adaptive and non-adaptive strategies, and on the problem of searching for optimal  finite strategies. We have also elucidated a connection between our notion of uncertainty function  and Bayesian decision theory.

\subsection{Related Work} In \cite{KB07}, K\"{o}pf and Basin
  introduced  an information-theoretic model of adaptive attackers for deterministic mechanisms. Their analysis is conducted essentially on the case of uniform prior distributions. Our model generalizes  \cite{KB07}  
  in several re\-spe\-cts: we consider probabilistic mechanisms, generic priors and generic uncertainty functions. More important than that, we contrast quantitatively  the efficiency of adaptive and non-adaptive strategies, we characterize   maximum leakage of infinite strategies,  and we show how to express   information leakage as a Bellman equation. The latter leads to search algorithms for optimal strategies that, when specialized to the deterministic case, are more time-efficient than the exhaustive  search outlined  in \cite{KB07} (see Section \ref{sec:finitestrat}.)

Our previous paper \cite{BP12} tackles 
multirun, non-adaptive adversaries, in the case of   min-entropy leakage. In this simpler setting, 
a special case of the present framework, one manages to obtain  simple analytical results, such as the exact convergence rate of the adversary's success probability as the number of observations goes to infinity.

Alvim et al. \cite{AAP11}
study information flow in interactive mechanisms, described
as probabilistic automata where  secrets and observables are seen as actions that
alternate during execution.  Information-theoretically, they
characterize these mechanisms as channels with feedback,
giving a Shannon-entropy based definition of leakage. Secret actions at each step   depend on  previous history, but it is not clear that this gives the adversary any ability to   adaptively influence the next observation, in our sense.

In \cite{ACPS12}, Alvim et al.  study $g$-\emph{leakage}, a generalization of min-entropy leakage, where the adversary's benefit deriving from a guess about a secret is specified using a \emph{gain function} $g$: intuitively, the closer the guess to the secret, the higher the gain. Alvim et al. derive general results about $g$-leakage, including bounds between min-capacity, $g$-capacity and Shannon capacity. Gain functions are conceptually very close to the proper scoring rules (PSRs) we considered in Section \ref{sec:roleconcave}. Abstracting  from the unimportant difference of encoding gains rather then  losses, a gain function can be seen in fact as a special case of a PSR where the forecast  put forward by the forecaster is always a Dirac's delta. One important technical  difference between \cite{ACPS12} on one side and the framework of PSRs and our paper on the other side, is that   entropy functions, as defined in \cite{ACPS12}, are all generalizations of the familiar min-entropy. As such, they are in general neither concave nor convex. A thorough investigation of the connections between $g$-leakage and PRSs is left for future work.
More or less contemporary to the short version of the present paper is Mardziel et al.'s  \cite{SSP14}, which extends the analysis via $g$-leakage functions to systems with memory.
  This work is similar in spirit to ours, but now successive responses to queries may not be independent, as the secret evolves over time. They too utilize backward induction to calculate leakage. Like in the static case, this dynamic $g$-leakage does not lend itself to be recast in the present framework. A dynamic approach is also at the core of a  model  in \cite{BPP11} based on Hidden Markov Models, where the observed system     evolves over time, although the secret is fixed.

After the publication of the short version of the present paper \cite{BP14}, we learned from a  statistician colleague about the existence of a large body of work in Bayesian forecasting  and PSRs. A good synthesis of this research can be found in   the works of DeGroot, see e.g. \cite{DeGroot62,DeGroot70} and references therein, although the terminology used there is slightly different (utility functions are considered rather than scoring rules.)  Remarkably, \cite{DeGroot62} contains considerations on sequential   observations and decisions which, despite the different terminology and emphases,  come  very close to our adaptive model of \qif. Dawid  \cite{Dawid} and Gneiting and Raftery \cite{GR07} give     modern accounts of these themes.
The role of concavity in \qif\ is also central to some recent works by McIver et al. \cite{LICS12,POST14}. An important result is the presentation of  a prior vulnerability    as a ``disorder test'' that is, interestingly, defined in terms of continuous and concave functions. It would be interesting to see how much these approaches share  with the one based on PSRs.

The use of Bellman equation and backward induction, applied to multi-threaded programs, in a context where strategies are schedulers,  is also found in \cite{Chen}.
\cite{advise1,advise2} propose    models  to assess system security against classes of adversaries characterized by   user-specified `profiles'. While these models share some similarities with ours - in particular,  they too employ \mdp's to keep track of possible adversary strategies -  their intent is quite different from ours: they are used to build and assess    analysis   tools, rather than to obtain analytical results. Also, the strategies they consider are tailored to worst-case adversary's utility, which, differently from our   average-case measures, is not apt to express information leakage.

\subsection{Further work}
There are several directions   worth being   pursued, starting from the present work. First, one would like to implement and experiment with the search algorithm described in Section \ref{sec:finitestrat}. Adaptive querying of datasets  might represent an ideal ground for evaluation of such algorithms.
Second, one would like to investigate worst-case variations of the present framework: an interesting possibility is to devise an adaptive version  of Differential Privacy \cite{D06,DMNS06} 
or   one of its variants \cite{BPa12}.
Finally, connections between (adaptive) \qif\ and PSRs in Statistics deserve to be further explored.

\section*{Acknowledgement}
\noindent  We should like to thank Fabio Corradi at { DiSIA}  for stimulating discussions on the relationship between information leakage and inference in Statistics, and in particular for pointing out to us the work on scoring rules in Bayesian decision theory.

We thank an anonymous reviewer of \cite{BP14} for suggesting us the envelopes example (Example \ref{ex:envelopes}).

\bibliographystyle{alpha}


\appendix
\ifmai \setcounter{lem}{0}
\renewcommand{\thelemma}{\Alph{section}\arabic{lem}}

\setcounter{thm}{0}
\renewcommand{\thetheorem}{\Alph{section}\arabic{thm}}

\setcounter{prop}{0}
\renewcommand{\theproposition}{\Alph{section}\arabic{prop}}

\setcounter{cor}{0}
\renewcommand{\thecorollary}{\Alph{section}\arabic{cor}}

\setcounter{rem}{0}
\renewcommand{\theremark}{\Alph{section}\arabic{rem}}
\fi

\section{Proofs of Lemma \ref{lemma:uno} and of Lemma \ref{lemma:due}}\label{app:lemmas}

\begin{lem}[Lemma \ref{lemma:uno}]  $I_\sigma(X;Y)\geq 0$. Moreover $I_\sigma(X;Y)= 0$ if $X$ and $Y$ are independent.
\end{lem}
\begin{proof}
First note that, by a simple manipulation relying on Bayes theorem, we can express the prior distribution $p(\cdot)$ as follows, for each $x\in \X$ (summation ranges over $y$ of positive probability):
\begin{eqnarray*}
p(x) & = & \sum_y p(y)p (x|y)\,.
\end{eqnarray*}
Hence by concavity of $U$ and Jensens's inequality:
\begin{eqnarray*}
U(X) & = & U(p)\\
     & = & U\left(\sum_y p(y)p (\cdot | y) \right)\\
     & \geq & \sum_y p(y) U(p(\cdot |y))\\
     & = & U(X|Y)\,.
\end{eqnarray*}
Moreover, if $X$ and $Y$ are independent, then for each $y$ of positive probability, $p(x)=p(x|y)$, hence $\sum_y p(y)p (\cdot|y)=p(\cdot)$, so that $U(X)=U(p)=U(X|Y)$.
\end{proof}

\begin{lem}[Lemma \ref{lemma:due}]
$x \equiv x'$ iff for each $a\in Act$, $p_a(\cdot|x) = p_a(\cdot|x')$. In other words, $\equiv$ is $\cap_{a\in Act}\equiv_a$.
\end{lem}
\begin{proof}
Let $x,x'\in \X$. Assume first $x\equiv_a x'$ for each action $a$. Let $\sigma$ be any finite strategy   and $y^ny_{n+1}$ a sequence such that $y^n\in\dom(\sigma)$ and $y^ny_{n+1}\notin\dom(\sigma)$. From \eqref{eq:probspace}, we know that
\begin{eqnarray*}
p_\sigma(y^ny_{n+1}|x) & = & p_{a_1}(y_1|x)\cdots p_{a_n}(y_n|x)p_{a_{n+1}}(y_{n+1}|x)\\
p_\sigma(y^ny_{n+1}|x') & = & p_{a_1}(y_1|x')\cdots p_{a_n}(y_n|x')p_{a_{n+1}}(y_{n+1}|x')
\end{eqnarray*}
where $a_1=\sigma(\epsilon)$, $a_2=\sigma(y_1)$,..., $a_n=\sigma(y_1\cdots y_{n-1})$ and $a_{n+1}=\sigma(y^n)$. From the above equations, it is immediate to conclude that,  as $p_a(\cdot|x)=p_a(\cdot|x')$ for each action $a$, then $p_\sigma(y^ny_{n+1}|x)=p_\sigma(y^ny_{n+1}|x')$. Since this holds for any finite $\sigma$, we conclude $x\equiv x'$. The other direction is obvious, as $p_a(\cdot|x)  =p_{[a]}(\cdot|x)$ (where $[a]$ is a length 1 strategy) for any $a$ and $x$.
\end{proof}

\section{Proof of Theorem \ref{th:due}}\label{app:proof}
In order to prove Theorem \ref{th:due}, we introduce some terminology and concepts from the information-theoretic \emph{method of types} \cite{CT06}.
For distributions $p(\cdot)$ and $q(\cdot)$ on $\Y$, we let their \emph{Kullback-Leibler (KL) divergence}
be defined as: $D(p\|q)\defin \sum_y p(y)\log\frac{p(y)}{q(y)}$, with the proviso that
$0\cdot \log \frac 0 {p(y)}= 0$ and $p(y)\cdot \log \frac {p(y)} 0= +\infty$ for $p(y)>0$.
Given $n\geq 1$, and a sequence $y^n\in \Y^n$,  the \emph{type} (or empirical distribution) of $y^n$, denoted $t_{y^n}$, is the probability distribution over $\Y$ defined thus: $t_{y^n}(y)\defin \frac{\mathrm{n}(y|y^n)}n$, where $\mathrm{n}(y|y^n)$ denotes the number of occurrences of $y$ in $y^n$.
In this section, $H$ will always stand for Shannon entropy, $H(p)=-\sum_x p(x)\log p(x)$.
We will often abbreviate $H(t_{y^n})$ as $H(y^n)$, and $D(t_{y^n}||q)$ as $D({y^n}||q)$, thus denoting the type
by a corresponding sequence, when no confusion arises.   Given $\varepsilon>0$ and a probability distribution $q(\cdot)$ on $\Y$, the ``ball'' of $n$-sequences whose type is   within distance $\varepsilon$ of $q(\cdot)$ is defined thus:
\begin{eqnarray*}
B^{(n)}(q,\varepsilon) & \defin & \{ y^n\,:\,D(y^n||q)\leq \varepsilon\}\,.
\end{eqnarray*}

\noindent
We shall also make use of the following new terminology about sequences. Assume $|Act|=k$. Given a sequence $y^n=(y_1,y_2,\ldots,y_n)$ and an integer $j=1,\ldots,k$, we shall denote by $y^n(j)$ the subsequence $(y_j,y_{k+j},y_{2k+j},\ldots)$, obtained by taking the symbols of $y^n$ at position  $j, k+j, 2k+j,\ldots$.
In the rest of the section, unless otherwise stated, we let $\sigma$ be the infinite non-adaptive strategy that plays actions $a_1,\ldots,a_k, a_1, a_2,\ldots$, in a lock-step fashion:  $\sigma(y^j)\defin a_{(j\bmod k)+1}$. For any $n\geq 1$, we let $\sigma_n$ be the truncation at level $n$ of $\sigma$: $\sigma_n\defin \sigma\backslash n$.
For a prior $p(\cdot)$, let  $p_{\sigma_n}$ be the resulting joint probability distribution on $\X\times \Y^n$: note that, for each $x$, the support of $p_{\sigma_n}(\cdot|x)$ is included in $\Y^n$. Let $(X,Y^n)$ be jointly distributed according to $p_{\sigma_n}$: here we have introduced the superscript ${}^n$ to record explicitly the dependence of $Y$ from $n$.
Let us define the set
of  sequences $y^n$   where the  type of each sub-sequence $y^n(i)$ is within $\varepsilon$ distance
of the distribution   $p_{a_i}(\cdot|x)$, thus:
\begin{eqnarray}\label{defUU}
\UU^{(n)}(x,\varepsilon) & \defin & \{y^n\,:\, D \left(\,y^n(i)||p_{a_i}(\cdot|x)\,\right)\leq \varepsilon\,\text{ for }i=1,\ldots,k\}\,.
\end{eqnarray}
\noindent
Furthermore, we define the following quantities depending on a given sequence $y^n$ and $x\in \X$:
\begin{eqnarray}
\HH(y^n) & \defin & \sum_{i=1}^k H(y^n(i))\label{defHH}
\end{eqnarray}
\begin{eqnarray}
\DD(y^n||p_{\sigma_n}(\cdot|x)) & \defin & \sum_{i=1}^k D(y^n(i)||p_{a_i}(\cdot|x))\label{defDD}\,.
\end{eqnarray}
\noindent
Finally, for each sequence $y^m\in\Y^m$ and for each action $a\in Act$, we let
\begin{eqnarray}\label{product}
  p^m_a(y^m|x) = \prod_{i=1}^{m}{p_{a}(y_i|x)}
\end{eqnarray}
\noindent
(this is the probability of generating $y^m$ with i.i.d. extractions obeying distribution $p_a(\cdot|x)$.)

\begin{lem}\label{lemma:KL}
Let $n$ be a multiple of $k$ and $x\in \X$. Then
\[
p_{\sigma_n}(y^n|x)=2^{-\frac n k [ \HH(y^n)+\DD(y^n||p_{\sigma_n}(\cdot|x))]}\,.
\]
\end{lem}
\begin{proof}
\begin{align}
p_{\sigma_n}(y^n|x) 
&= \prod_{j=1}^k{p^{\frac n k}_{a_j}(y^n(j)|x)}\label{eq_A}\\
&= \prod_{j=1}^k 2^{-\frac{n}{k}\left(H(y^n(j))+D\left(y^n(j)\|p_{a_j}(\cdot|x)\right)\right)}\label{eq_B}\\
&= 2^{-\frac{n}{k}\sum_{j=1}^k{\left(H(y^n(j))+D\left(y^n(j)\|p_{a_j}(\cdot|x)\right)\right)}}\nonumber\\
&= 2^{-\frac{n}{k}\left(\HH(y^n)+\DD\left(y^n\|p_{\sigma_n}(\cdot|x)\right)\right)}\label{eq_C}
\end{align}
\noindent
where:  (\ref{eq_A}) follows from re-arranging   factors and the definition of $p^{\frac n k}_{a_j}(\cdot)$;
(\ref{eq_B}) follows from   \cite[Theorem~11.1.2]{CT06};
in (\ref{eq_C}) we have applied definitions (\ref{defHH}) and (\ref{defDD}).
\end{proof}

Below, for a set $A$ and a distribution $q(\cdot)$, we let $q(A)$ denote $\sum_{a\in A} q(a)$.

\begin{lem}\label{lemma:ball}
Let $n$ be a multiple of $k$, $x\in \X$ and $\varepsilon>0$.
Then
\[
p_{\sigma_n}(\UU^{(n)}(x,\varepsilon)|x)\geq 1-2^{-\frac n k \varepsilon}{\left(\frac n k +1\right)}^{k|\Y|}C
\]
\noindent
for some constant $C$, not depending on $n$.
\end{lem}
\begin{proof}
Let $m=n/k$. We give a lower bound on the probability of 
$\UU^{(n)}(x,\varepsilon)$ as follows.
\begin{eqnarray}
p_{\sigma_n}(\UU^{(n)}(x,\varepsilon)|x) & = & \sum_{y^n\in \UU^{(n)}(x,\varepsilon) } p_{\sigma_n}(y^n|x)\ = \  \sum_{y^n\in \UU^{(n)}(x,\varepsilon) } \prod_{i=1}^k p^m_{a_i}(y^n(i)|x)\nonumber\\
& = &  \prod_{i=1}^k \sum_{y^m\in B^{(m)}(p^m_{a_i}(\cdot|x),\varepsilon)} p_{a_i}(y^m|x) \ = \ \prod_{i=1}^k   p^m_{a_i}\left(B^{(m)}(p_{a_i}(\cdot|x),\varepsilon)\,|\,x\right)\label{eq:tricky}\\
    & \geq  &  \prod_{i=1}^k   \left(1- 2^{-m\varepsilon}(m+1)^{|\Y|}\right) \quad\  = \ \left(1- 2^{-m\varepsilon}(m+1)^{|\Y|}\right)^k\label{eq:tricky2}\\
       & = & 1+ \sum_{i=1}^k {k \choose i}(-1)^i 2^{-mi\varepsilon}(m+1)^{|\Y|i}\ \geq \ 1 - 2^{-m\varepsilon}(m+1)^{k|\Y|}C\label{eq:tricky3}
\end{eqnarray}
\noindent
where: the first equality in \eqref{eq:tricky} follows from the definition of $\UU^{(n)}(x,\varepsilon)$; the inequality \eqref{eq:tricky2} follows
 from   \cite[Eq. 11.67]{CT06};  in \eqref{eq:tricky3}, $C=k\cdot \max_i{k \choose i}$
  (note that ${k \choose i}$ is maximum when $i=\lceil k/2 \rceil $.)
\end{proof}

\begin{lem}\label{lemma:2epsilon}
Let $x,x'\in \X$, with $x\not\equiv x'$. Let $n\geq 1$. Then there is $\varepsilon >0$ such that
$\UU^{(n)}(x,2\varepsilon)\cap \UU^{(n)}(x',2\varepsilon)=\es$.
\end{lem}
\begin{proof}
It is well-known that given any two distinct distributions $p(\cdot)$ and $q(\cdot)$,   there is $\varepsilon>0$ such that $B^{(n)}(p,2\varepsilon)\cap B^{(n)}(q,2\varepsilon)=\es$ (this is a consequence of Pinsker's inequality,   \cite[Lemma 11.6.1]{CT06}.)  Thus, choose $\varepsilon >0$ such that,  for some $j$, $B^{(n)}\left(p_{a_j}(\cdot|x),2\varepsilon\right)\cap B^{(n)}\left(p_{a_j}(\cdot|x'),2\varepsilon\right)=\es$: the wanted statement follows from the definition of $\UU^{(n)}(x,2\varepsilon)$ and $\UU^{(n)}(x',2\varepsilon)$.
\end{proof}

\ifmai
By contradiction, suppose that such $\varepsilon$ does not exist. It means that for any $\varepsilon\geq 0$ there is at least a value $y^{n,\varepsilon}$ belonging to the intersection $\UU^{(n)}(x,2\varepsilon)\cap \UU^{(n)}(x',2\varepsilon.)$
Let $\delta\triangleq\min_{y^n\in\Y^n}{\|p_\sigma(\cdot|x)-p_\sigma(\cdot|x')\|}_1.$  Then:
\begin{align}
\delta &\ \leq\  {\|p_\sigma(\cdot|x)-p_\sigma(\cdot|x')\|}_1 \ = \ \sum_{y\in\Y^n}{\left|p_\sigma(y^n|x)-p_\sigma(y^n|x')\right|}\nonumber\\
&\ =\  \prod_{i=1}^k{{\|p_{a_i}(\cdot|x)-p_{a_i}(\cdot|x')\|}_1}\label{algebra}\\
&\ \leq\ \prod_{i=1}^k{\left({\|y^{n,\varepsilon}(i)-p_{a_i}(\cdot|x)\|}_1+{\|y^{n,\varepsilon}(i)-p_{a_i}(\cdot|x')\|}_1\right)}\nonumber\\
&\ \leq\ \prod_{i=1}^k{\left(\sqrt{2\ln{2}D(y^{n,\varepsilon}(i)\|p_{a_i}(\cdot|x))}+
\sqrt{2\ln{2}D(y^{n,\varepsilon}(i)\|p_{a_i}(\cdot|x'))}\right)}\label{Pinker}\\
&\ \leq\ \prod_{i=1}^k{(2\sqrt{4\varepsilon\ln{2}})}=4^k{(\varepsilon\ln{2})}^{\frac{k}{2}}\nonumber
\end{align}
where: $\ln$ denotes the natural logarithm; (\ref{algebra}) is due to some algebra; for the last inequality we have exploited the hypothesis
that $y^{n,\varepsilon}\in\UU^{(n)}(x,2\varepsilon)\cap \UU^{(n)}(x',2\varepsilon)$ and definition (\ref{defUU});
(\ref{Pinker}) follows from
Lemma $11.6.1$ in \cite{CT06}, for which for any two distributions $p$ and $q$ on
the same set, it holds true that $D(p\|q)\geq\frac{1}{2\ln{2}}{\|p-q\|}_1^2.$
If now we choose $\varepsilon<\frac{1}{16}{(\delta\ln{2})}^{\frac{2}{k}}$ we obtain a contradiction and so we have the thesis.
\mik{Francesca, dovrebbe essere sul modello di quella già usata in altri paperi, basata sulla Pinsker inequality,  ma qui tenendo conto della differente definizione di $\UU$. Puoi vedere i dettagli?}
\end{proof}
\fi

We are now set to prove Theorem \ref{th:due}.

\begin{thm}[Theorem \ref{th:due}] There is a total, non-adaptive strategy $\sigma$ s.t. $I_\sigma(\St,p)=
 I(X;[X])$. Consequently, $I_\star(\St,p)=  I(X\,;\,[X])$.
\end{thm}
\begin{proof}
Using the notation previously introduced, we shall prove that, as $n\rightarrow \infty$
\begin{eqnarray}\label{eq:limit}
U(X|Y^n)\longrightarrow L \quad \quad \text{ for some }L\leq U(X|[X])\,.
 \end{eqnarray}
This will imply the thesis, as then $I_\sigma(\St,p)\geq I(X;[X])$, which, by virtue of Theorem \ref{th:uno}, implies $I_\sigma(\St,p)= I(X;[X])$.

Let the equivalence classes of $\equiv$ be $c_1,\ldots,c_K$. For each $i=1,\ldots,K$,
choose a representative $x_i\in c_i$ of nonzero probability (if it exists; otherwise
class $c_i$ is just discarded.)  We can compute as follows.
\begin{eqnarray}
U(X|Y^n) & = & \sum_{y^n,x} p(x)p_{\sigma_n}(y^n|x)U(p_{\sigma_n}(\cdot|y^n))\nonumber\\
        & = & \sum_x p(x) \sum_{y^n} p_{\sigma_n}(y^n|x)U(p_{\sigma_n}(\cdot|y^n))\nonumber\\
        & \leq & \sum_x p(x)U\left(\sum_{y^n} p_{\sigma_n}(y^n|x)p_{\sigma_n}(\cdot|y^n)\right)\label{eq:JJ}\\
        &= & \sum_{c_i}\sum_{x\in c_i}p(x)U\left(\sum_{y^n} p_{\sigma_n}(y^n|x_i)p_{\sigma_n}(\cdot|y^n)\right)\nonumber\\
        &= & \sum_{c_i} p(c_i)U\underbrace{\left(\sum_{y^n} p_{\sigma_n}(y^n|x_i)p_{\sigma_n}(\cdot|y^n)\right)}_{\defin q^n_i(\cdot)}\nonumber\\[-8pt]
        & = & \sum_{c_i}p(c_i)U(q^n_i)\label{eq:qi}
\end{eqnarray}
where the inequality in \eqref{eq:JJ} stems from $U$'s concavity and Jensen's  inequality.
We will now show that there is a sub-sequence of indices $\{n_j\}$ such that for each $i=1,\ldots,K$,
\begin{eqnarray}\label{eq:limit2}
q^{n_j}_i(\cdot) & \longrightarrow & p(\cdot|c_i)
\end{eqnarray}
(according to any chosen metrics  in $\P$.) This will imply \eqref{eq:limit}: in fact,
  by virtue of $U$'s continuity, we will have, on the chosen sub-sequence,
  $\sum_{c_i}p(c_i)U(q^{n_j}_i)\rightarrow \sum_{c_i} p(c_i)U(p(\cdot|c_i)) = U(X|[X])$.
 Hence, by virtue of \eqref{eq:qi},   on the chosen sub-sequence  and hence on every sequence,
  we will have $U(X|Y^n)\rightarrow L \leq U(X|[X])$, which is \eqref{eq:limit}.

In order to prove \eqref{eq:limit2}, take any $n\geq 1$ that is a multiple of $k$,
  and choose any $\varepsilon>0$ such that $\UU^{(n)}(x,2\varepsilon)\cap \UU^{(n)}(x',2\varepsilon)=\es$
   whenever $x\not\equiv x'$ (the existence of such an  $\varepsilon$ is guaranteed by Lemma
    \ref{lemma:2epsilon}.) Consider a generic $x\in c_i$ such that $p(x)>0$.   We have the following
    lower bound for $q^n_i(x)$.
\begin{eqnarray}
q^n_i(x) & = & \sum_{y^n}\frac{p_{\sigma_n}(y^n|x_i)p_{\sigma_n}(y^n|x_i)p(x)}
{\sum_{x'}p_{\sigma_n}(y^n|x')p(x')}\  = \ \sum_{y^n} \frac{p_{\sigma_n}(y^n|x_i)}{\frac{p(c_i)}{p(x)} + \sum_{x'\not\equiv x_i}\frac {p_{\sigma_n}(y^n|x')p(x')}{p_{\sigma_n}(y^n|x_i)p(x_i)} }\label{eq:bayes}\\
& \geq & \sum_{y^n\in \UU^{(n)}(x,\varepsilon)} \frac{p_{\sigma_n}(y^n|x_i)}{\frac{p(c_i)}{p(x)} + \sum_{x'\not\equiv x_i} 2^{-\frac n k [\DD(y^n||p_{\sigma_n}(\cdot|x'))-\DD(y^n||p_{\sigma_n}(\cdot|x))  ]}\frac {p(x')}{p(x)} }\label{eq:DK1}\\
& \geq & \sum_{y^n\in \UU^{(n)}(x,\varepsilon)} \frac{p_{\sigma_n}(y^n|x_i)}{\frac{p(c_i)}{p(x)} + \sum_{x'\not\equiv x_i} 2^{-  n  \varepsilon}\frac {p(x')}{p(x)} }\label{eq:DK2}\\
& = &  p_{\sigma_n}\left( \,\UU^{(n)}(x_i,\varepsilon)  \,|\,x_i\right) \frac{1}{ \frac{p(c_i)}{p(x)} + 2^{-n\varepsilon}C'}\label{eq:DK3}\\
& \geq  &  \frac{ 1-2^{- \frac n k\varepsilon}C(\frac n k +1)^{k|\Y|}   }{ \frac{p(c_i)}{p(x)} + 2^{-n\varepsilon}C'}\label{eq:DK4}
\end{eqnarray}

\noindent
where: \eqref{eq:bayes} follows from the definition of $q^n_i(x)$ and an
application of Bayes rule, and from the fact that $p_{\sigma_n}(y^n|x)=p_{\sigma_n}(y^n|x_i)$; \eqref{eq:DK1} follows from a simple union bound and from Lemma \ref{lemma:KL};
 \eqref{eq:DK2} follows from the fact that, by assumption,
 $\UU^{(n)}(x,2\varepsilon)\cap \UU^{(n)}(x',2\varepsilon)=\es$
 (also note that $\UU^{(n)}(x,2\varepsilon)=  \UU^{(n)}(x_i,2\varepsilon)$);
 \eqref{eq:DK3} follows by definition of $\UU^{(n)}(x,\varepsilon)= \UU^{(n)}(x_i,\varepsilon)$;
 here $C'$ is a suitable constant, not depending on $n$; \eqref{eq:DK4} follows from Lemma \ref{lemma:ball}.

Now, let $\{n_j\}$ be a sequence of indices such that, for each $x$ and $i$,   $q^{n_j}_i(x)$
converges to a limit, say $L_i(x)$ (such a sub-sequence must exist, by Bolzano-Weierstrass.)
The inequality
\[
q^n_i(x)\geq \frac{ 1-2^{-\frac  n k\varepsilon}C(\frac n k +1)^{k|\Y|}   }{ \frac{p(c_i)}{p(x)} + 2^{-n\varepsilon}C'}
\]
\noindent
which holds for each $n$ that is a multiple of $k$, implies that these limits satisfy
$L_i(x)\geq \frac {p(x)}{p(c_i)}$. Since point-wise convergence for each $x$ implies convergence of $q^{n_j}_i(\cdot)$   to a probability distribution, we have that, for each $i$ and $x$, actually equality must hold: $L_i(x)= \frac {p(x)}{p(c_i)}$.
Thus, for each $i=1,\ldots,K$,  $q^{n_j}_i(\cdot)\rightarrow p(\cdot|c)$,   which proves
\eqref{eq:limit2}.
\end{proof}

\section{Backward Induction} \label{app:backward}
We give some details of the derivation of an optimal strategy of length 2 for the system in Example \ref{exDB}.
We take   Shannon entropy as the chosen uncertainty measure.
In Figure  \ref{fig:exMDP} we give a partial representation of the related \mdp.
According to the Backward Induction method, we   compute
the reward of each node, starting from the leaves of the \mdp-tree and then propagating them up to the root. Let us denote by $R(n)$ the reward assigned to a node $n$. Applying the algorithm, we compute as follows (recall that levels are counted from the root, which has level 0).
\begin{itemize}
\item  Level 4. The reward associated to each leaf node $n$ is $R(n)=0$
\item  Level 3. Each probabilistic node, which is the $a$-child of a decision node $h$, receives as a reward the sum of its immediate gain, $I_{[a]}(\St,p^h)$, and the average reward of its children:
 $R(4)=R(5)=\log{5}-\frac{2}{5}-\frac{3}{5}\log{3}$, $R(6)=R(8)=0$, $R(7)=\frac{1}{3}$, $R(9)=\log{3}-\frac{4}{3}$, and so on.
\item Level 2. Each decision node receives as  a reward the maximum of its children' rewards (and the corresponding action is recorded):
  $R(B)=\log{5}-\frac{2}{5}-\frac{3}{5}\log{3}$,  $R(C)=\frac{1}{3}$, $R(D)=\log{3}-\frac{4}{3}$ and so on.
\item Level 1.
Each probabilistic node receives the following rewards: $R(1)=\log{10}-\frac{9}{25}-\frac{27}{50}\log{3}\approx 2.11$,
 $R(2)=\log{10}-\frac{11}{15}-\frac{1}{5}\log{3}\approx 2.27$
and $R(3)=\log{10}-\frac{3}{5}-\frac{1}{5}\log{3}\approx 2.4$.
\item Level 0. The decision node $A$ receives the reward $R(A) = R(3)\approx 2.4$
\end{itemize}
Taking into account the maximal children' rewards selected at each decision node, we have the following optimal strategy of length 2 (its tree representation is given in Figure \ref{optimalstrategy}.)
\[
\sigma\triangleq\left\{
\begin{array}{llllll}
\varepsilon\mapsto Age &29\mapsto ZIP &30\mapsto Date &31\mapsto ZIP &32\mapsto Date &64\mapsto Date\\
65\mapsto ZIP &66\mapsto ZIP &67\mapsto Date &68\mapsto Date &69\mapsto Date. &\
\end{array}
\right.
\]

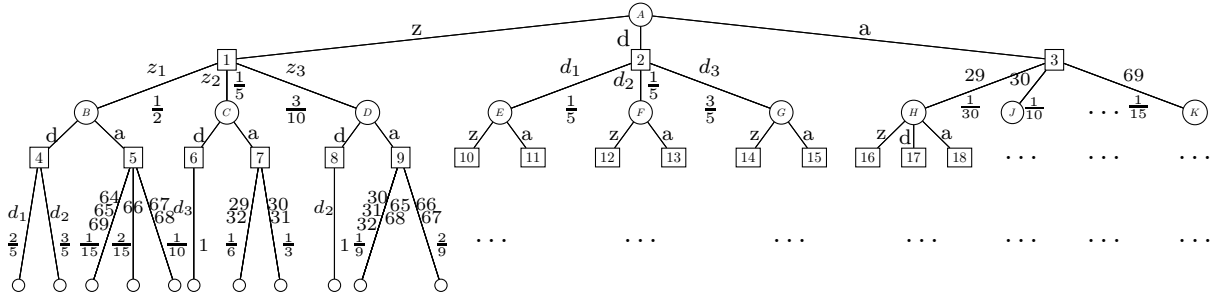
\begin{figure}
\tikzstyle{bag} = [text centered,circle, minimum width=3pt,fill, inner sep=0pt]
\tikzstyle{Node}=[draw, minimum width=1.5pt,fill, inner sep=1.5pt]
{\centering
\begin{tikzpicture}[xscale=1.1,
  level 1/.style={sibling distance=5cm, level distance=0.6cm},
  level 2/.style={sibling distance=1.7cm, level distance=0.7cm},
  level 3/.style={sibling distance=0.8cm, level distance=0.6cm},
  level 4/.style={sibling distance=0.5cm, level distance=1.7cm}]

\node[circle,draw, scale=0.45] (a) {\footnotesize $A$}
child { node[draw, scale=0.55] (1) {1}
child { node[circle,draw, scale=0.45] (b) {\footnotesize$B$} child { node [draw, scale=0.55, left=3pt] (4) {4}
      child { node [circle,draw, scale=0.45] (r) {}} child { node [circle, draw, scale=0.45] (s) {} }}
      child { node [draw, scale=0.55,right=3pt] (5) {5}
      child { node [circle,draw, scale=0.45] (t) {}} child { node [circle, draw, scale=0.45] (u) {} }
      child { node [circle,draw, scale=0.45] (v) {}}}}
child { node[circle,draw, scale=0.45] (c) {\footnotesize$C$} child { node [draw, scale=0.55] (6) {6}
      child { node [circle,draw, scale=0.45] (w) {}}}
      child { node [draw, scale=0.55] (7) {7}
      child { node [circle,draw, scale=0.45] (x) {}}
      child { node [circle,draw, scale=0.45] (z) {}}}}
child { node[circle,draw, scale=0.45] (d) {\footnotesize$D$} child { node [draw, scale=0.55] (8) {8}
      child { node [circle,draw, scale=0.45] (a1) {}}}
      child { node [draw, scale=0.55] (9){9}
      child { node[left=5pt] [circle,draw, scale=0.45] (b1) {}}
      child { node[right=5pt] [circle,draw, scale=0.45] (d1) {}}}}}
child { node[draw, scale=0.55] (2) {2}
child {node[circle,draw, scale=0.45] (e) {\footnotesize$E$} child {node[draw, scale=0.5] (10) {10}} child {node[draw, scale=0.5] (11) {11}}}
child {node[circle,draw, scale=0.45] (f) {\footnotesize$F$} child {node[draw, scale=0.5] (12) {12}} child {node[draw, scale=0.5] (13) {13}}}
child {node[circle,draw, scale=0.45] (g) {\footnotesize$G$} child {node[draw, scale=0.5] (14) {14}} child {node[draw, scale=0.5] (15) {15}}}}
child { node[draw, scale=0.55] (3) {3} child {node[circle,draw, scale=0.45] (h) {\footnotesize$H$} child {node[draw, scale=0.5, right=6pt] (16) {16}} child {node[draw, scale=0.5] (17) {17}} child {node[draw, scale=0.5, left=6pt] (18) {18}}} child {node[circle,draw, scale=0.45, left=0.9cm] (j) {\footnotesize$J$}} child {node[circle,draw, scale=0.45] (k) {\footnotesize$K$}}};
\draw (5.6cm,-1.3cm) node {$\ldots$};
\draw (5.6cm,-1.9cm) node {$\ldots$};
\draw (6.7cm,-1.9cm) node {$\ldots$};
\draw (4.6cm,-1.9cm) node {$\ldots$};
\draw (-1.8cm,-3cm) node {$\ldots$};
\draw (0cm,-3cm) node {$\ldots$};
\draw (1.8cm,-3cm) node {$\ldots$};
\draw (3.4cm,-3cm) node {$\ldots$};
\draw (4.6cm,-3cm) node {$\ldots$};
\draw (5.6cm,-3cm) node {$\ldots$};
\draw (6.7cm,-3cm) node {$\ldots$};

\path[]
    (a) edge node [above=3pt, left] {\footnotesize z} (1)
    (a) edge node [left] {\footnotesize d} (2)
    (a) edge node [above=3pt, right] {\footnotesize a} (3)
    (1) edge node [above] {\scriptsize $z_1$ } (b)
    (1) edge node [below] {\scriptsize $\frac{1}{2}$ } (b)
    (1) edge node [above=2pt, left=-2pt] {\scriptsize $z_2$ } (c)
    (1) edge node [right=-2pt] {\scriptsize $\frac{1}{5}$ } (c)
    (1) edge node [above] {\scriptsize $z_3$ } (d)
    (1) edge node [below] {\scriptsize $\frac{3}{10}$ } (d)
    (2) edge node [above] {\scriptsize $d_1$ } (e)
    (2) edge node [below] {\scriptsize $\frac{1}{5}$ } (e)
    (2) edge node [above=2pt, left=-2pt] {\scriptsize $d_2$ } (f)
    (2) edge node [right=-2pt] {\scriptsize $\frac{1}{5}$ } (f)
    (2) edge node [above] {\scriptsize $d_3$ } (g)
    (2) edge node [below] {\scriptsize $\frac{3}{5}$ } (g)
    (3) edge node [above=4pt, left=-4pt] {\tiny $29$ } (h)
    (3) edge node [below=8pt, right=-14pt] {\tiny $\frac{1}{30}$ } (h)
    (3) edge node [above=3pt, left=-3pt] {\tiny $30$} (j)
    (3) edge node [below=8pt, right=-8pt] {\tiny $\frac{1}{10}$ } (j)
    (3) edge node [above=4pt, right=-4pt] {\tiny $69$ } (k)
    (3) edge node [below=8pt,left=-14pt] {\tiny $\frac{1}{15}$ } (k)
    (b) edge node [left=-2pt] {\scriptsize d} (4)
    (b) edge node [right=-2pt] {\scriptsize a} (5)
    (c) edge node [left=-2pt] {\scriptsize d} (6)
    (c) edge node [right=-2pt] {\scriptsize a} (7)
    (d) edge node [left=-2pt] {\scriptsize d} (8)
    (d) edge node [right=-2pt] {\scriptsize a} (9)
    (e) edge node [left=-2pt] {\scriptsize z} (10)
    (e) edge node [right=-2pt] {\scriptsize a} (11)
    (f) edge node [left=-2pt] {\scriptsize z} (12)
    (f) edge node [right=-2pt] {\scriptsize a} (13)
    (g) edge node [left=-2pt] {\scriptsize z} (14)
    (g) edge node [right=-2pt] {\scriptsize a} (15)
    (h) edge node [left=-2pt] {\scriptsize z} (16)
    (h) edge node [left=-3pt] {\scriptsize d} (17)
    (h) edge node [right=-2pt] {\scriptsize a} (18)
    (4) edge node [above=-8pt, left=-1pt] {\tiny $\frac{2}{5}$ } (r)
    (4) edge node [above=4pt, left=-4pt] {{\tiny $d_1$}} (r)
    (4) edge node [above=-8pt, right=-1pt] {\tiny $\frac{3}{5}$ } (s)
    (4) edge node [above=4pt, right=-4pt] {\tiny $d_2$ } (s)
    (5) edge node [above=10pt, left=-7pt] {{\tiny $64$}} (t)
    (5) edge node [above=5pt, left=-5pt] {{\tiny $65$}} (t)
    (5) edge node [above=0pt, left=-3pt] {{\tiny $69$}} (t)
    (5) edge node [above=-8pt, left] {\tiny $\frac{1}{15}$} (t)
    (5) edge node [above] {\tiny $66$} (u)
    (5) edge node [above=-8pt, left=-4pt] {\tiny $\frac{2}{15}$} (u)
    (5) edge node [above=-8pt, right] {\tiny $\frac{1}{10}$} (v)
    (5) edge node [above=7.5pt, right=-6pt] {{\tiny $67$}} (v)
    (5) edge node [above=2.5pt, right=-4pt] {{\tiny $68$}} (v)
    (6) edge node [above=6pt, left=-4pt] {\tiny $d_3$} (w)
    (6) edge node [above=-8pt, right=-2pt] {\tiny $1$} (w)
    (7) edge node [above=-8pt, left] {\tiny $\frac{1}{6}$} (x)
    (7) edge node [above=7.5pt, left=-4pt] {{\tiny $29$}} (x)
    (7) edge node [above=2.5pt, left=-3pt] {{\tiny $32$}} (x)
    (7) edge node [above=-8pt, right] {\tiny $\frac{1}{3}$} (z)
    (7) edge node [above=7.5pt, right=-5pt] {{\tiny $30$}} (z)
    (7) edge node [above=2.5pt, right=-4pt] {{\tiny $31$}} (z)
    (8) edge node [above=6pt, left=-4pt] {\tiny $d_2$} (a1)
    (8) edge node [above=-8pt, right=-2pt] {\tiny $1$} (a1)
    (9) edge node [above=10pt, left=-7pt] {{\tiny $30$}} (b1)
    (9) edge node [above=5pt, left=-5pt] {{\tiny $31$}} (b1)
    (9) edge node [above=0pt, left=-3pt] {{\tiny $32$}} (b1)
    (9) edge node [above=7.2pt, right=0pt] {{\tiny $65$}} (b1)
    (9) edge node [above=2.2pt, right=-2pt] {{\tiny $68$}} (b1)
    (9) edge node [above=-8pt, left=1pt] {\tiny $\frac{1}{9}$} (b1)
    (9) edge node [above=-8pt, right=1pt] {\tiny $\frac{2}{9}$} (d1)
    (9) edge node [above=7.5pt, right=-6pt] {{\tiny $66$}} (d1)
    (9) edge node [above=2.5pt, right=-4pt] {{\tiny $67$}} (d1)
    ;
\end{tikzpicture}
}
\caption{{\small A partial representation of the \mdp\  for Example \ref{exDB}. 
 Here,  the symbols `z', `d'  and `a' are abbreviations for   actions $ZIP$, $Date$ and $Age$, respectively. For simplicity, probability distributions labelling  decision nodes are not shown.
  Moreover, leaves   with  the same labels and the same father,  and the corresponding incoming arcs, have been coalesced.
}}
\label{fig:exMDP}
\end{figure}
\end{document}